\documentclass[11pt,a4paper]{article}

\IfFileExists{lmodern.sty}{%
  \usepackage[T1]{fontenc}
  \usepackage{lmodern}%
}{\usepackage[OT1]{fontenc}}
\usepackage[utf8]{inputenc}
\usepackage[margin=1in]{geometry}
\usepackage{amsmath,amssymb,amsthm,mathtools}
\usepackage{bm}
\usepackage{booktabs}
\usepackage{graphicx}
\usepackage{microtype}
\usepackage{enumitem}
\usepackage[numbers,sort&compress]{natbib}
\usepackage{xcolor}
\usepackage{hyperref}
\usepackage{aliascnt}
\usepackage[nameinlink,noabbrev]{cleveref}

\hypersetup{
  colorlinks=true,
  linkcolor=blue!55!black,
  citecolor=teal!55!black,
  urlcolor=magenta!55!black,
  pdftitle={Two-Regime Risk Summaries under Convex Loss: Stability and Computation},
  pdfauthor={Delia Coculescu; Adriana Nistor; Ionel Popescu},
  pdfkeywords={boundary atoms, binary scalar quantization,
    convex loss, cutoff stability, epi-convergence, psi-weak topology,
    empirical consistency, risk summary}
}

\theoremstyle{plain}
\newtheorem{theorem}{Theorem}[section]
\newaliascnt{proposition}{theorem}
\newtheorem{proposition}[proposition]{Proposition}
\aliascntresetthe{proposition}
\newaliascnt{lemma}{theorem}
\newtheorem{lemma}[lemma]{Lemma}
\aliascntresetthe{lemma}
\newaliascnt{corollary}{theorem}
\newtheorem{corollary}[corollary]{Corollary}
\aliascntresetthe{corollary}
\theoremstyle{definition}
\newaliascnt{definition}{theorem}
\newtheorem{definition}[definition]{Definition}
\aliascntresetthe{definition}
\newaliascnt{assumption}{theorem}
\newtheorem{assumption}[assumption]{Assumption}
\aliascntresetthe{assumption}
\newtheorem{algorithm}{Algorithm}[section]
\theoremstyle{remark}
\newaliascnt{remark}{theorem}
\newtheorem{remark}[remark]{Remark}
\aliascntresetthe{remark}
\newaliascnt{example}{theorem}
\newtheorem{example}[example]{Example}
\aliascntresetthe{example}

\crefname{theorem}{theorem}{theorems}
\Crefname{theorem}{Theorem}{Theorems}
\crefname{proposition}{proposition}{propositions}
\Crefname{proposition}{Proposition}{Propositions}
\crefname{lemma}{lemma}{lemmas}
\Crefname{lemma}{Lemma}{Lemmas}
\crefname{corollary}{corollary}{corollaries}
\Crefname{corollary}{Corollary}{Corollaries}
\crefname{definition}{definition}{definitions}
\Crefname{definition}{Definition}{Definitions}
\crefname{assumption}{assumption}{assumptions}
\Crefname{assumption}{Assumption}{Assumptions}
\crefname{remark}{remark}{remarks}
\Crefname{remark}{Remark}{Remarks}
\crefname{example}{example}{examples}
\Crefname{example}{Example}{Examples}
\crefname{algorithm}{algorithm}{algorithms}
\Crefname{algorithm}{Algorithm}{Algorithms}

\newcommand{\R}{\mathbb R}
\newcommand{\Rbar}{\overline{\mathbb R}}
\newcommand{\E}{\mathbb E}
\newcommand{\1}{\mathbf 1}
\newcommand{\Law}{\mathcal L}
\newcommand{\cP}{\mathcal P}
\newcommand{\cM}{\mathcal M}
\newcommand{\cV}{\mathcal V}
\newcommand{\cC}{\mathcal C}
\newcommand{\cO}{\mathcal O}
\newcommand{\argmin}{\operatorname*{arg\,min}}
\newcommand{\dist}{\operatorname{dist}}
\newcommand{\Var}{\operatorname{Var}}
\newcommand{\supp}{\operatorname{supp}}

\newcommand{\PKlimsup}{\mathop{\mathrm{Limsup}}}
\newcommand{\psito}{\xrightarrow{\psi}}
\newcommand{\W}{\mathsf W}
\newcommand{\restr}[2]{#1\!\restriction_{#2}}
\newcommand{\backlink}[1]{\par\smallskip\noindent\emph{\hyperref[#1]{Return to the statement.}}}

\title{Two-Regime Risk Summaries under Convex Loss:\\
Stability and Computation}
\author{%
Mihaela-Adriana Nistor\thanks{University of Bucharest, Bucharest, Romania. Email: \href{mailto:mihaelaadriana.nistor@gmail.com}{mihaelaadriana.nistor@gmail.com}.}
\and
Ionel Popescu\thanks{University of Bucharest, Bucharest, Romania; Institute of Mathematics of the Romanian Academy, Bucharest, Romania. Emails: \href{mailto:ionel.popescu@fmi.unibuc.ro}{ionel.popescu@fmi.unibuc.ro} and \href{mailto:ionel.popescu@imar.ro}{ionel.popescu@imar.ro}.}%
}
\date{}

\begin{document}
\maketitle

\begin{abstract}
We study a two-regime summary of a real-valued loss distribution.  The two
representative levels and the boundary between them are chosen by minimizing
a convex residual loss.  When the distribution has an atom at the boundary,
assigning that atom to the lower or upper regime can give different optimized
costs.  Taking the better whole-atom assignment yields the
lower-semicontinuous cutoff profile.  We prove epi-convergence of this profile
under a loss-adapted \(\psi\)-weak topology and obtain convergence of minimum
values and outer stability of optimal cutoff sets.  Under a unique finite
atom-free limiting cutoff, positive regime masses, and unique conditional
centers, the two fitted levels and a canonical one-jump representation
converge as well.  We also examine the resulting
summary from the viewpoint of monetary risk axioms, establish explicit
failures of external monotonicity and subadditivity, and give finite-support
algorithms with conditional error control.  A four-scenario model illustrates
the cutoff geometry and the role of ties and atoms.
\end{abstract}



\section{Introduction}\label{sec:intro}

Monetary risk measures summarize a loss distribution by one capital amount.
This scalar viewpoint is well suited to ranking, allocation, and regulation,
and is supported by the established theory of coherent and convex risk
measures \citep{Artzner1999,FoellmerSchied2002,Kusuoka2001}.  In applications
where moderate and severe losses have different operational meanings,
however, it can also be useful to retain a boundary between two severity
classes and one representative value within each class.  The object studied
here is such a distributional summary: the law chooses both representative
levels and the boundary separating them.

Let \(X\) be a real-valued loss with law \(\mu\), and let \(G\) be a convex
penalty for residual error.  The closed-left version of the problem is
\begin{equation}\label{eq:intro-problem}
 \inf_{a,b,c\in\R}
 \E\!\left[
   G(X-a)\1_{\{X\le c\}}+G(X-b)\1_{\{X>c\}}
 \right].
\end{equation}
For a fixed cutoff, \(a\) and \(b\) are the optimal \(G\)-centers of the two
restricted laws.  Optimizing over \(c\) then gives a nondecreasing two-level
encoding of the realized loss.  We call the resulting object a
\emph{two-regime risk summary}.  Here ``regime'' refers only to two static
severity cells; it does not posit a latent state or a regime-switching time
series.  The report consists of a cutoff and two representative severities,
rather than a single amount of eligible capital, and it is set-valued when
several cuts or centers are equally good.

The choice of \(G\) determines the meaning of the representative levels.
Squared loss gives conditional means, absolute and check losses give medians
and quantiles, asymmetric squared loss gives expectiles, and Huber loss gives
a robust location
\citep{Huber1964,KoenkerBassett1978,NeweyPowell1987}.
Thus the construction belongs naturally to the expected-loss approach used
in forecast scoring and in the risk-quadrangle formulation
\citep{Gneiting2011,RockafellarUryasev2013}.  It can be read as replacing one
unconditional location by two conditional locations whose partition is also
selected from the law.

From approximation theory, the same optimization is binary scalar
quantization with the translation-invariant distortion
\(d_G(x,y)=G(x-y)\).  Convexity orders the two reconstruction levels and
reduces the search over arbitrary measurable cells to a half-line threshold.
For quadratic loss the levels are the two principal points of the
distribution \citep{Flury1990,Tarpey1994,LiFlury1995}; the same criterion
appears in optimal grouping, piecewise-constant approximation, and Otsu
thresholding \citep{Eubank1988,MeaseNair2006,Otsu1979}.  The Lloyd--Max
conditions and the general theory of optimal quantization provide the basic
approximation framework \citep{Max1960,Lloyd1982,GrayNeuhoff1998,GrafLuschgy2000}.
For ordered discrete data, the corresponding prefix formulations are closely
related to one-dimensional clustering algorithms
\citep{Fisher1958,WangSong2011,NielsenNock2014}.  Least-squares stump
estimation also fits two levels and a split point, although its inferential
questions usually concern a smooth regression function and a continuous
design rather than convergence of a distributional cutoff profile at a
charged boundary \citep{BuehlmannYu2002,BanerjeeMcKeague2007}.

Throughout, we interpret these approximation concepts in the language of a
risk model: codepoints become severity levels, the Voronoi boundary becomes a
reported cutoff, and perturbations of the source law describe changes in, or
approximations to, the loss model.

The present work is also connected to two recent risk-oriented formulations.
\citet{NistorPopescu2026} study the two-step convex-loss problem, establish
the half-line reduction, and, for quadratic loss, investigate uniqueness
under log-concavity and a higher-dimensional elliptical extension.  Here we
retain the one-dimensional threshold representation and concentrate instead
on stability of the cutoff, its behavior at atoms, and its interpretation as
a risk summary.  \citet{FaugerasPages2024} use constrained two-point
quantization to construct a magnitude--propensity risk summary and also
discuss unconstrained two-point compression.  Their principal construction
fixes one codepoint at zero, while both levels are endogenous in
\eqref{eq:intro-problem}; their quadratic analysis also records possible mass
at a threshold.  These perspectives place the present summary within the
existing use of quantization for multicomponent descriptions of risk.

Stability of optimal codebooks under changing laws has a substantial
literature, beginning with consistency results for quadratic \(k\)-means and
extending to more general distortions, algorithms, and set-valued limits
\citep{Pollard1981,Pollard1982,AbayaWise1984,SabinGray1986,LiuPages2020,
Jaffe2025}.  General variational results for integral objectives with changing
measures and discontinuous integrands supply a broader framework
\citep{ArtsteinWets1988,ArtsteinWets1994,VogelLachout2003,
FeinbergKasyanovRoyset2023,RockafellarWets1998}.  For the cutoff coordinate,
the discontinuity has a particular form.  If \(\mu\) charges \(c\), the
closed-left and open-left thresholds produce two different restricted laws;
the center within each restriction is then reoptimized.  A sequence of nearby
discrete cuts can even converge to a fractional allocation of the limiting
boundary mass, although the admissible limiting action must assign that atom
wholly to one regime.

The specific novelty of the paper is the resolution of this boundary-atom
problem in the cutoff coordinate.  We take the better of the two whole-atom
assignments and show that the resulting objective is the greatest
lower-semicontinuous minorant of the raw cutoff profile.  Under the weighted
weak topology associated with
\(\psi(x)=1+G((1+\eta)x)\), these relaxed profiles epi-converge.  The key
limiting argument uses concavity of the optimized one-regime value as a
function of the underlying measure: a fractional limiting allocation cannot
improve on both whole-atom endpoint assignments.  It follows that minimum
values converge and that cluster points of exact or asymptotically optimal
cutoffs remain optimal.  When the population cutoff is unique, finite,
atom-free, and separates two positive-mass cells with unique centers, the
cutoff, both levels, and a canonical one-jump graph converge jointly.  The
same analysis yields consistency for deterministic finite-support and
empirical approximations.  This topology retains the tail integrability
needed by an unbounded residual loss
\citep{KratschmerSchiedZaehle2012,KratschmerSchiedZaehle2014}.

The risk interpretation leads to a complementary question: which properties
of scalar risk functionals remain meaningful for this regime-valued output?
We formulate the answer at the level of optimizer correspondences and encoded
losses.  Law invariance, cash equivariance, positive-scale equivariance for
homogeneous losses, and monotonicity of each fitted encoding in the realized
loss are retained.  In contrast, reoptimizing the partition can destroy
external monotonicity and subadditivity; explicit finite examples show that
the latter can fail even for comonotone losses.  For finite-support laws we
connect the analysis to an exhaustive ordered-prefix scan, provide routines
for several convex losses, and give a conditional objective-gap and
minimizer-retention bound.  A four-scenario model makes cutoff gaps, tied
partitions, and duplicate atoms visible and supplies a reproducible test case.

The paper proceeds as follows.  \Cref{sec:framework} defines the summary and
establishes its structural properties.  \Cref{sec:risk-axioms} studies its
transformation and order properties.  \Cref{sec:stability} develops the
cutoff-profile and optimizer stability theory, and \Cref{sec:approx} applies
it to deterministic and empirical approximation.  \Cref{sec:computation}
gives finite-support algorithms and error control, while
\Cref{sec:fourpoint} presents the four-scenario example.  The final section
collects the interpretation, limitations, and possible extensions.

\section{Mathematical framework}\label{sec:framework}

\subsection{Notation and elementary probability concepts}\label{sec:notation}

Write \(\cP(\R)\) for the Borel probability laws on \(\R\), and
\(\cM_+(\R)\) for the finite nonnegative Borel measures.  For
\(\nu\in\cM_+(\R)\) and a Borel set \(B\), the restriction
\(\restr{\nu}{B}\) is defined by
\((\restr{\nu}{B})(A)=\nu(A\cap B)\).  Weak convergence
\(\nu_n\Rightarrow\nu\) means convergence of
\(\int f\,d\nu_n\) to \(\int f\,d\nu\) for every bounded continuous
\(f\).  We use the compactified line
\(\Rbar=\R\cup\{-\infty,+\infty\}\), with its usual order topology;
thus \(c_n\to+\infty\) has its literal meaning.

For probability laws \(\mu,\nu\), first-order stochastic order
\(\mu\preceq_{\mathrm{st}}\nu\) means
\(\mu((t,\infty))\le\nu((t,\infty))\) for every \(t\), or equivalently
that there is a coupling \(X\sim\mu,Y\sim\nu\) with \(X\le Y\) almost
surely.  We use \(u_+=\max(u,0)\) and \(u_-=\max(-u,0)\).

\subsection{One-regime centers}

\begin{assumption}\label{ass:G}
Unless stated otherwise, \(G:\R\to[0,\infty)\) satisfies:
\begin{enumerate}[label=(G\arabic*),leftmargin=2.5em]
\item \(G\) is finite and convex;
\item \(G(0)=0\), and \(0\) is the unique minimizer of \(G\);
\item \(G\) is coercive: \(G(u)\to\infty\) as \(|u|\to\infty\).
\end{enumerate}
\end{assumption}

A finite convex function on \(\R\) is continuous and has finite left and
right derivatives \(G'_{-}\le G'_{+}\) at every point.  The normalization
in (G2) is natural for residual error.  Uniqueness is needed for the
support-hull and regime-separation assertions below.  Flat-bottom losses
are discussed in \Cref{rem:flat-bottom}.

The following elementary bound will be used repeatedly:
\begin{equation}\label{eq:convex-domination}
 G(x-m)
 \le
 \frac{1}{1+\eta}G((1+\eta)x)
 +\frac{\eta}{1+\eta}
  G\!\left(-\frac{1+\eta}{\eta}m\right).
\end{equation}
It follows because \(x-m\) is the convex combination of
\((1+\eta)x\) and \(-((1+\eta)/\eta)m\), with weights
\(1/(1+\eta)\) and \(\eta/(1+\eta)\).  In particular, translated losses
are dominated by one common integrable envelope while \(m\) ranges over a
compact set.

For \(\nu\in\cM_+(\R)\), define
\begin{align}
 \cV_G(\nu)
   &:=\inf_{m\in\R}\int G(x-m)\,\nu(dx), \label{eq:VG}\\
 \cM_G(\nu)
   &:=\argmin_{m\in\R}\int G(x-m)\,\nu(dx),
      \qquad \nu\ne0, \label{eq:MG}
\end{align}
and set \(\cV_G(0)=0\).  Multiplying a nonzero measure by a positive
constant leaves \(\cM_G\) unchanged, so unnormalized restrictions can be
used without repeatedly forming conditional laws.

\begin{proposition}[Existence and interval structure]\label{prop:center-existence}
Suppose
\(\int G((1+\eta)x)\,\nu(dx)<\infty\) for some \(\eta>0\).
If \(\nu\ne0\), then \(\cM_G(\nu)\) is a nonempty compact interval.
If \(G\) is strictly convex, it is a singleton.  Moreover
\[
 \cM_G(\nu)\subseteq
 \bigl[\inf\supp(\nu),\,\sup\supp(\nu)\bigr],
\]
with the usual interpretation when the support is unbounded.
\end{proposition}

\noindent The proof is deferred to
\hyperref[app:center]{\Cref*{app:center}}.

\begin{remark}[Flat-bottom losses]\label{rem:flat-bottom}
The unique-minimum clause in (G2) is substantive.  For example, if
\(G(u)=\dist(u,[-1,1])^2\) and \(\nu=\delta_0\), then
\(\cM_G(\nu)=[-1,1]\), which is not contained in
\(\operatorname{conv}(\supp\nu)=\{0\}\).  More generally, when
\(\argmin G=[\ell,r]\), the existence and continuity results remain valid,
but the support-hull bound expands to
\[
 [\inf\supp(\nu)-r,\ \sup\supp(\nu)-\ell].
\]
All named losses in \Cref{tab:centers}, with parameters in their stated
open ranges, have the unique minimizer \(0\).
\end{remark}

For \(t\in\R\), write \(F_\nu(t):=\nu((-\infty,t])\).  For
finite-support measures, the center condition can be written without
measure-valued subdifferentials.  A point \(m\) is a center if and only if
\begin{equation}\label{eq:subgradient-center}
 \int G'_{-}(x-m)\,\nu(dx)\le0
 \le\int G'_{+}(x-m)\,\nu(dx).
\end{equation}
If \(G\) is differentiable, both bounds coincide and the condition is
\(\int G'(x-m)\,\nu(dx)=0\).  For general measures the same statements
hold whenever subdifferentiation and integration can be interchanged, for
example under an integrable local envelope for the one-sided derivatives;
none of the general convergence proofs below relies on this score formula.

\begin{table}[t]
\centering
\caption{One-regime centers induced by standard residual losses.}
\label{tab:centers}
\begin{tabular}{lll}
\toprule
\(G(u)\) & center & defining condition\\
\midrule
\(u^2\) & mean & \(m=\int x\,d\nu/\nu(\R)\)\\
\(|u|\) & median set & at most half the mass strictly on either side\\
\(\tau u_++(1-\tau)u_-\) & \(\tau\)-quantile set
  & \(F_\nu(m-)\le\tau\nu(\R)\le F_\nu(m)\)\\
\(\tau u_+^2+(1-\tau)u_-^2\) & \(\tau\)-expectile
  & weighted first-moment balance\\
Huber \(G_\delta\) & Huber center & clipped-score balance\\
\bottomrule
\end{tabular}
\end{table}

\subsection{From conditional centers to cutoff profiles}

Let \(\mu\in\cP(\R)\) satisfy the integrability condition of
\Cref{prop:center-existence}.  For \(c\in\R\), set
\[
\begin{array}{ll}
 \mu_c^{\le,-}:=\restr{\mu}{(-\infty,c]},&
 \mu_c^{\le,+}:=\restr{\mu}{(c,\infty)},\\[2mm]
 \mu_c^{<,-}:=\restr{\mu}{(-\infty,c)},&
 \mu_c^{<,+}:=\restr{\mu}{[c,\infty)}.
\end{array}
\]
The superscript records whether a possible atom at \(c\) belongs to the
lower or upper regime.
At the endpoints we use the common convention
\[
 \mu_{+\infty}^{\sigma,-}=\mu,\quad
 \mu_{+\infty}^{\sigma,+}=0,\qquad
 \mu_{-\infty}^{\sigma,-}=0,\quad
 \mu_{-\infty}^{\sigma,+}=\mu,
 \qquad \sigma\in\{\le,<\}.
\]

\begin{definition}[Cutoff profiles]\label{def:profiles}
The closed-left and open-left raw profiles are
\begin{align}
 H_G^\le(\mu,c)
  &:=\cV_G(\mu_c^{\le,-})+\cV_G(\mu_c^{\le,+}),\label{eq:Hle}\\
 H_G^<(\mu,c)
  &:=\cV_G(\mu_c^{<,-})+\cV_G(\mu_c^{<,+}).\label{eq:Hlt}
\end{align}
The boundary-relaxed profile is
\begin{equation}\label{eq:Hrel}
 \underline H_G(\mu,c)
 :=\min\{H_G^\le(\mu,c),H_G^<(\mu,c)\}.
\end{equation}
At \(c=\pm\infty\), all three profiles are defined to equal
\(\cV_G(\mu)\).
\end{definition}

At an atom-free cutoff, the two raw profiles agree.  At an atom,
\eqref{eq:Hrel} assigns the whole boundary mass to whichever regime gives
the smaller optimized cost.  It does \emph{not} split one atom internally.
The relaxation removes an arbitrary choice between \((-\infty,c]\) and
\((-\infty,c)\), and will be seen in \Cref{sec:epi} to be the
epi-stable profile.

\begin{proposition}[Conditional centers]\label{prop:conditional-centers}
For a fixed \(c\), every minimizer of
\[
 \inf_{a,b\in\R}\int
 \left[
   G(x-a)\1_{\{x\le c\}}+G(x-b)\1_{\{x>c\}}
 \right]\mu(dx)
\]
satisfies
\[
 a\in\cM_G(\mu_c^{\le,-}),\qquad
 b\in\cM_G(\mu_c^{\le,+})
\]
whenever the corresponding measure is nonzero.  The analogous statement
holds for the open-left convention.  Under strict convexity, both active
levels are unique.
\end{proposition}

\begin{proof}
For fixed \(c\), the objective is the sum
\[
 \int_{(-\infty,c]}G(x-a)\,\mu(dx)
 +\int_{(c,\infty)}G(x-b)\,\mu(dx).
\]
The variables separate, and each active variable must minimize its own
summand.  Strict convexity invokes
\Cref{prop:center-existence}.
\end{proof}

\subsection{Why two arbitrary regimes reduce to a threshold}

One could start with an arbitrary Borel set \(A\) and the approximation
\(a\1_A+b\1_{A^c}\).  Convexity forces an ordered representative.  This
one-dimensional reduction is part of the two-step construction in
\citet{NistorPopescu2026}; we include the short argument to fix the boundary
conventions needed later.

\begin{proposition}[Half-line reduction (recalled)]\label{prop:halfline}
Fix \(a<b\).  The map
\[
 x\longmapsto G(x-a)-G(x-b)
\]
is nondecreasing.  Consequently, among all Borel sets \(A\), a minimizer of
\[
 \int\!\left[G(x-a)\1_A(x)+G(x-b)\1_{A^c}(x)\right]\mu(dx)
\]
can be chosen as a lower half-line.  If the two losses tie on an interval,
either complete assignment of the tie interval is optimal, so an open-left
or closed-left threshold representative exists.  The same conclusion
holds for \(a>b\) after interchanging labels.
\end{proposition}

\begin{proof}
Put \(d=b-a>0\) and \(u=x-b\).  The displayed difference is
\(G(u+d)-G(u)\).  Increments of a convex function over an interval of fixed
length are nondecreasing in the base point.  The objective equals
\[
 \int G(x-b)\,\mu(dx)
 +\int_A\!\bigl[G(x-a)-G(x-b)\bigr]\,\mu(dx).
\]
It is minimized pointwise by including the negative-difference set and
excluding the positive-difference set.  The strict sublevel set of a
nondecreasing function is a lower half-line, and its zero set is an
interval.  Assigning the entire zero interval to either side gives one of
the two stated boundary conventions.
\end{proof}

\begin{corollary}[Equivalence with two-state approximation]\label{cor:two-state}
Let
\[
 Q_G(\mu):=\inf_{\substack{a,b\in\R\\A\in\mathcal B(\R)}}
 \int\!\left[G(x-a)\1_A(x)+G(x-b)\1_{A^c}(x)\right]\,\mu(dx).
\]
Then
\[
 Q_G(\mu)
 =\inf_{c\in\Rbar}H_G^\le(\mu,c)
 =\min_{c\in\Rbar}\underline H_G(\mu,c).
\]
\end{corollary}

\begin{proof}
If \(a=b\), every assignment has the same integrand and any threshold is a
valid representative.
For \(a\ne b\), \Cref{prop:halfline} shows that, for every fixed pair, a
threshold assignment is at least as good as any Borel assignment.
Conversely, threshold assignments are Borel assignments.  Minimizing the
active levels at each threshold gives the relaxed profile.  The equality
with the raw closed-left infimum follows from
\(H_G^<(\mu,c)=H_G^\le(\mu,c-)\); see
\Cref{prop:profile-regularity}.  The same proposition gives lower
semicontinuity and attainment of the relaxed minimum.  Thus the relaxation
preserves the optimal value, although it can add a limiting action at an
atom.
\end{proof}

\subsection{Optimizer correspondences}

\begin{definition}[Optimal cuts, actions, and encodings]\label{def:omega}
Define
\begin{equation}\label{eq:cutset}
 \cC_G(\mu):=\argmin_{c\in\Rbar}\underline H_G(\mu,c).
\end{equation}
For an action \(\theta=(a,b,c,\sigma)\), with
\(\sigma\in\{\le,<\}\), write
\[
 q_\theta(x)=
 \begin{cases}
  a\1_{\{x\le c\}}+b\1_{\{x>c\}},&\sigma=\le,\\
  a\1_{\{x<c\}}+b\1_{\{x\ge c\}},&\sigma=<,
 \end{cases}
\]
using the extended-order convention when \(c=\pm\infty\).
The action correspondence \(\mathcal A_G(\mu)\) consists of all such
\(\theta\) for which \(c\in\cC_G(\mu)\), the convention \(\sigma\) attains
\eqref{eq:Hrel}, and each active level belongs to the corresponding center
set.  If one cell is empty, its inactive level is set equal to the active
level.  Finally,
\[
 \cO_G(\mu):=\{q_\theta:\theta\in\mathcal A_G(\mu)\}
\]
is the associated encoding correspondence.
\end{definition}

The set-valued formulation is essential.  Atoms, support gaps, nonunique
centers, or equal split costs can all cause nonuniqueness.  A deterministic
tie-break produces a representative, but it need not vary continuously with
the law unless the limiting optimizer is unique.  The set \(\cO_G(\mu)\) is an optimizer
correspondence of severity encodings; it is not an upper set of acceptable
capital portfolios in the sense of set-valued monetary risk measures.  The
natural risk-measure axioms for this output are formulated in
\Cref{sec:risk-axioms}.

\subsection{Ordering and equivariance}

For a nonzero finite measure \(\nu\), write
\[
 m_G^-(\nu):=\min\cM_G(\nu),\qquad
 m_G^+(\nu):=\max\cM_G(\nu).
\]

\begin{proposition}[Stochastic monotonicity of centers]\label{prop:center-order}
If probability laws \(\mu\preceq_{\mathrm{st}}\nu\) satisfy the standing
integrability assumptions, then
\[
 m_G^-(\mu)\le m_G^-(\nu),\qquad
 m_G^+(\mu)\le m_G^+(\nu).
\]
\end{proposition}

\begin{proof}
Couple \(X\sim\mu,Y\sim\nu\) so that \(X\le Y\) almost surely.  If
\(m_1>m_2\), convexity makes
\[
 z\longmapsto G(z-m_1)-G(z-m_2)
\]
nonincreasing.  Therefore
\begin{equation}\label{eq:center-crossing}
 \E[G(X-m_1)-G(X-m_2)]
 \ge \E[G(Y-m_1)-G(Y-m_2)].
\end{equation}
If \(x=m_G^-(\mu)>y=m_G^-(\nu)\), optimality makes the left side of
\eqref{eq:center-crossing}, with \((m_1,m_2)=(x,y)\), nonpositive and the
right side nonnegative.  Both must be zero.  Then \(y\) is also a
\(\mu\)-center, contradicting the minimality of \(x\).  The proof for the
maximal centers is identical: equality would make \(x\) a \(\nu\)-center
larger than \(m_G^+(\nu)\).
\end{proof}

\begin{corollary}[Ordered and cutoff-monotone levels]\label{cor:conditional-order}
Whenever both regimes at \(c\) have positive mass,
\[
 m_G^+(\mu_c^{\le,-})\le c\le m_G^-(\mu_c^{\le,+}).
\]
Both endpoints of the lower and upper conditional center sets are
nondecreasing functions of \(c\).  If
\(\cM_G(\mu)=\{m_G(\mu)\}\), then any conditional-center selections satisfy
\[
 a(c)\to m_G(\mu)\quad(c\to+\infty),\qquad
 b(c)\to m_G(\mu)\quad(c\to-\infty).
\]
\end{corollary}

\noindent A proof is given in
\hyperref[app:conditional-order]{\Cref*{app:conditional-order}}.

\begin{proposition}[Translation and positive-scale equivariance]
\label{prop:equivariance}
Let \(T_t\mu\) be the law of \(X+t\).  Then
\[
 \underline H_G(T_t\mu,c+t)=\underline H_G(\mu,c),\qquad
 \cC_G(T_t\mu)=\cC_G(\mu)+t,
\]
and an optimal cut-level triple \((a,b,c)\) becomes
\((a+t,b+t,c+t)\).

If \(G(su)=s^\gamma G(u)\) for every \(s>0\), and \(S_s\mu\) is the law of
\(sX\), then
\[
 \underline H_G(S_s\mu,sc)=s^\gamma\underline H_G(\mu,c),\qquad
 \cC_G(S_s\mu)=s\cC_G(\mu),
\]
with levels \((sa,sb)\).
\end{proposition}

\begin{proof}
The substitutions \(m'=m-t\) and \(m=su\), respectively, give
\[
 \cV_G(T_t\nu)=\cV_G(\nu),\quad
 \cM_G(T_t\nu)=\cM_G(\nu)+t,
\]
and
\[
 \cV_G(S_s\nu)=s^\gamma\cV_G(\nu),\quad
 \cM_G(S_s\nu)=s\cM_G(\nu).
\]
Translation and positive scaling carry every half-line restriction to its
corresponding restriction at the transformed cutoff.  Applying the
identities to both profile terms proves the claims.
\end{proof}

\section{Transformation properties and limits of monetary interpretation}
\label{sec:risk-axioms}

We interpret the two-regime summary through classical risk-measure axioms.
Optimal quantization has been used as a multicomponent risk summary
\citep{FaugerasPages2024}; for the endogenous-level formulation considered
here, we determine which monetary and coherent-type properties are preserved.
We use the loss convention: larger
random variables mean worse outcomes, so a scalar monetary risk measure
\(\rho\) is monotone when \(X\le Y\) implies
\(\rho(X)\le\rho(Y)\), and cash additive when
\(\rho(X+t)=\rho(X)+t\).  Because the present output is not scalar, its
order and aggregation rule must first be made explicit.  This is also
necessary in abstract-cone, vector-valued, and set-valued theories
\citep{JouiniMeddebTouzi2004,CascosMolchanov2007,HamelHeyde2010}.

\subsection{Actions and encoded losses}

For a loss \(X\) satisfying the standing integrability condition, let
\(\mathcal A_X:=\mathcal A_G(\Law(X))\) be its action correspondence from
\Cref{def:omega}.  For \(\theta\in\mathcal A_X\), write
\begin{equation}\label{eq:encoded-loss}
 \widehat X_\theta:=q_\theta(X).
\end{equation}
Thus actions and deterministic encodings are functionals of the law, while
\(\widehat X_\theta\) is a random variable on the probability space carrying
\(X\).
It is useful to distinguish \emph{internal monotonicity}, meaning that a
single fitted encoding does not assign a smaller level to a larger realized
loss, from \emph{external monotonicity}, which compares outputs fitted
separately to two ordered losses.

\begin{proposition}[Retained axioms and state order]\label{prop:retained-axioms}
The two-regime correspondence has the following properties.
\begin{enumerate}[label=(\roman*),leftmargin=2.2em]
\item It is law invariant.
\item It is cash equivariant.  If \(t\in\R\), then
  \[
   \mathcal A_G(\Law(X+t))=\mathcal A_G(\Law(X))+t,
  \]
  where \(\theta+t=(a+t,b+t,c+t,\sigma)\), with the natural convention at
  \(c=\pm\infty\).  Corresponding encodings satisfy
  \(\widehat{X+t}_{\,\theta+t}=\widehat X_\theta+t\) almost surely.
\item If \(G(su)=s^\gamma G(u)\) for \(s>0\), then, for every \(s>0\),
  the action correspondence is positive-scale equivariant:
  \[
   \mathcal A_G(\Law(sX))=s\mathcal A_G(\Law(X)),\qquad
   \widehat{sX}_{\,s\theta}=s\widehat X_\theta ,
  \]
  where \(s\theta=(sa,sb,sc,\sigma)\).
\item Every \(\theta\in\mathcal A_G(\Law(X))\) has \(a\le b\), and its
  deterministic encoding is nondecreasing.  Moreover
  \[
   \operatorname*{ess\,inf}X
   \le \widehat X_\theta
   \le \operatorname*{ess\,sup}X
   \quad\text{almost surely},
  \]
  with the usual extended-real interpretation.
\item If \(X\) has at most two support points, then every optimal encoding
  is exact on the loss: \(\widehat X_\theta=X\) almost surely.
\end{enumerate}
\end{proposition}

\begin{proof}
Law invariance follows because the profile, center sets, and optimal actions
are defined from \(\Law(X)\).  The translation and scaling statements are
\Cref{prop:equivariance}, including the corresponding transformation of the
boundary convention.  For two nonempty regimes,
\Cref{cor:conditional-order} gives \(a\le c\le b\); an empty-regime output
is constant by definition.  The support-hull part of
\Cref{prop:center-existence} places each active level in the closed hull of
its regime support, proving the range bound and internal monotonicity.
Finally, a law with at most two support points admits a two-level encoding
of cost zero.  Since \(0\) is the unique minimizer of \(G\), every
zero-cost optimum agrees with \(X\) almost surely.
\end{proof}

\begin{corollary}[Additivity along an affine ray]\label{cor:ray-additivity}
Suppose \(G\) is positively homogeneous as in
\Cref{prop:retained-axioms}.  Choose the encodings compatibly with
translation and positive scaling.  Here \(\widehat Z\) denotes the chosen
optimal encoding of a loss \(Z\).  For \(\lambda>0\) and \(t\in\R\),
\[
 \widehat{\lambda X+t}=\lambda\widehat X+t.
\]
The same identity holds at \(\lambda=0\) at the encoding level, because a
constant loss is reproduced exactly.
In particular, when \(Y=\lambda X+t\),
\begin{equation}\label{eq:ray-additivity}
 \widehat{X+Y}=\widehat X+\widehat Y .
\end{equation}
\end{corollary}

\begin{proof}
For \(\lambda>0\), the first identity is
\Cref{prop:retained-axioms}(ii)--(iii).  Apply it to
\(Y=\lambda X+t\) and to \(X+Y=(1+\lambda)X+t\) to obtain
\eqref{eq:ray-additivity}.  If \(\lambda=0\), then \(Y=t\) has the constant
encoding \(t\), and translation equivariance gives the same conclusion.
\end{proof}

This is additivity along a common affine risk ray, not general comonotone
additivity.  The counterexample in \Cref{prop:subadd-fails} is comonotone
and shows that the distinction is substantive.

\begin{proposition}[Order and additivity on a fixed partition]
\label{prop:fixed-partition}
Let \(A\) be an exogenously fixed event with
\(\mathbb P(A),\mathbb P(A^c)>0\), and summarize a loss by one \(G\)-center
on \(A\) and one on \(A^c\).
\begin{enumerate}[label=(\roman*),leftmargin=2.2em]
\item If \(X\le Y\) almost surely, then the minimal and maximal conditional
center selections for \(X\) do not exceed their counterparts for \(Y\), in
either regime.
\item If \(G(u)=u^2\), then the two conditional means are additive:
\[
 \bigl(\E[X+Y\mid A],\E[X+Y\mid A^c]\bigr)
 =
 \bigl(\E[X\mid A],\E[X\mid A^c]\bigr)
 +
 \bigl(\E[Y\mid A],\E[Y\mid A^c]\bigr).
\]
\end{enumerate}
\end{proposition}

\begin{proof}
Conditioning the inequality \(X\le Y\) on either common event gives
first-order stochastic order of the two conditional laws, so (i) follows
from \Cref{prop:center-order}.  Part (ii) is linearity of conditional
expectation.
\end{proof}

Thus neither conditioning nor the block centers themselves cause the
failures below.  They arise because \(X\), \(Y\), and \(X+Y\) generally
select different endogenous threshold partitions.

\begin{proposition}[Mean-preserving scalar shadow]\label{prop:mean-shadow}
Let \(G(u)=u^2\) and \(X\in L^2\).  Every optimal regime encoding preserves
the unconditional mean:
\[
 \E\widehat X_\theta=\E X,\qquad
 \theta\in\mathcal A_G(\Law(X)).
\]
Consequently the probability-weighted barycenter of the two reported levels
is the monotone, cash-additive, and additive scalar risk functional
\(\E X\), independently of the selected optimal cutoff.
\end{proposition}

\begin{proof}
On every nonempty regime, the quadratic center is the corresponding
conditional mean.  Applying the law of total expectation to the two regime
events gives \(\E\widehat X_\theta=\E X\).  If one regime is empty, the encoding
is the unconditional mean and the same identity holds.
\end{proof}

\begin{remark}[Subadditivity on a common partition]
Suppose a single-valued one-regime center \(m_G\) is subadditive on every
conditional probability space, and the same event \(A\) is optimal for
\(X\), \(Y\), and \(X+Y\).  Applying subadditivity under
\(\mathbb P(\,\cdot\mid A)\) and
\(\mathbb P(\,\cdot\mid A^c)\) gives
\[
 a_{X+Y}\le a_X+a_Y,\qquad b_{X+Y}\le b_X+b_Y.
\]
For quadratic loss these are equalities by
\Cref{prop:fixed-partition}; upper expectiles provide a subadditive
one-regime example for \(\tau\in[1/2,1)\), their coherent parameter range
\citep{BelliniKlarMuellerRosazza2014}.  This result is deliberately
regime-aligned: endogenous repartitioning invalidates its premise.
\end{remark}

\subsection{Bayes actions and the inactive-coordinate caveat}

For a report \(\theta=(a,b,c,\sigma)\), define
\[
 I_{\le,c}(x)=\1_{\{x\le c\}},\qquad
 I_{<,c}(x)=\1_{\{x<c\}},
\]
and
\begin{equation}\label{eq:joint-score}
 S_G(\theta,x)
 =G(x-a)I_{\sigma,c}(x)
  +G(x-b)\bigl(1-I_{\sigma,c}(x)\bigr).
\end{equation}
On its natural integrability class, let
\begin{equation}\label{eq:bayes-action}
 \mathcal B_G(\mu):=
 \argmin_\theta\int S_G(\theta,x)\,\mu(dx)
\end{equation}
be the raw Bayes-action correspondence.  If one cell is empty, its level
does not enter the score and is therefore unidentified.  Let
\(\operatorname{can}_\mu\) replace such an inactive level by the active one,
as in \Cref{def:omega}.  The correct relation is
\[
 \mathcal A_G(\mu)
 =\operatorname{can}_\mu\!\bigl(\mathcal B_G(\mu)\bigr),
\]
not a literal equality between \(\mathcal A_G(\mu)\) and the raw argmin in
\eqref{eq:bayes-action}.  The score is selectively consistent for
\(\mathcal B_G\).  On a specified class of laws with two active cells, a
fixed boundary convention, and a unique Bayes report, it is strictly
consistent for the resulting point-valued action functional
\citep{FisslerZiegel2016}.  In general, inactive coordinates, boundary
labels at atom-free cutoffs, and cost ties prevent strict identification of
an arbitrary selector; exhaustive elicitation of the whole optimizer set is
a different question \citep{FisslerFrongilloHlavinovaRudloff2021}.

\subsection{Failure of external monotonicity}

A natural regime-valued analogue of scalar monotonicity would require that
ordered losses admit ordered reports,
\[
 X\le Y\quad\Longrightarrow\quad
 \widehat X\le\widehat Y\quad\text{almost surely}.
\]
A weaker coordinate test asks only that the fitted lower and upper levels
do not decrease.  Both versions fail.  The example below has unique optimal
partitions and unique levels, so neither a cutoff representative nor a
tie-breaking rule can repair it.

\begin{proposition}[Global monotonicity fails]\label{prop:monotonicity-fails}
Let \(G(u)=u^2\) and let the three states be equiprobable.  Set
\[
 \begin{array}{c|ccc}
       &1&2&3\\ \hline
 X     &0&1&3\\
 Y     &0&2&3
 \end{array}
 \qquad\text{so that }X\le Y.
\]
The unique optimal support partition of \(X\) is
\(\{0,1\}\mid\{3\}\), with levels \((1/2,3)\), whereas the unique optimal
support partition of \(Y\) is
\(\{0\}\mid\{2,3\}\), with levels \((0,5/2)\).  Hence both optimized
levels decrease under the deterioration \(X\le Y\), and the encoded losses
are not ordered.
\end{proposition}

\begin{proof}
For \(X\), the costs of the two nontrivial threshold partitions are
\[
 \frac{(0-\tfrac12)^2+(1-\tfrac12)^2}{3}=\frac16,
 \qquad
 \frac{(1-2)^2+(3-2)^2}{3}=\frac23.
\]
Thus \(\{0,1\}\mid\{3\}\) is the unique optimal partition.  For \(Y\),
the same two costs occur in the opposite order:
\[
 \frac{(2-\tfrac52)^2+(3-\tfrac52)^2}{3}=\frac16,
 \qquad
 \frac{(0-1)^2+(2-1)^2}{3}=\frac23.
\]
The corresponding reported values are
\[
 \widehat X=(1/2,1/2,3),\qquad
 \widehat Y=(0,5/2,5/2).
\]
The desired inequality fails in both the first and third states.  The full
closed cutoff gaps do not affect any value on the support.
\end{proof}

This does not contradict \Cref{prop:center-order}.  Within either law, the
two conditional centers straddle the cutoff and their extremal selections
are nondecreasing as that law's cutoff moves; see
\Cref{cor:conditional-order}.  Comparing \(X\) and \(Y\) after conditioning
at the same numerical value would use different events and is not covered
by that result.  External conditional order is guaranteed on a common
exogenous event, as in \Cref{prop:fixed-partition}.  Here the failure occurs
because the globally optimal support partition changes.

\subsection{Failure of subadditivity and the remaining bound}

For reports on one probability space, the direct analogue of scalar
subadditivity is
\[
 \widehat{X+Y}\le\widehat X+\widehat Y
 \quad\text{almost surely}.
\]
At coordinate level this would at least require the upper severity level of
\(X+Y\) not to exceed the sum of the two upper levels.  Even this fails for
comonotone nonnegative losses.

\begin{proposition}[Subadditivity fails, even comonotonically]
\label{prop:subadd-fails}
Let \(G(u)=u^2\) on three equiprobable states and set
\[
 \begin{array}{c|ccc}
       &1&2&3\\ \hline
 X     &0&0&2\\
 Y     &0&2&3\\
 X+Y   &0&2&5 .
 \end{array}
\]
The variables \(X\) and \(Y\) are comonotone.  Their unique optimal support
partitions induce the level pairs
\[
 (0,2),\qquad (0,5/2),\qquad (1,5)
\]
for \(X,Y,X+Y\), respectively.  Consequently
\[
 5>2+\frac52,
\]
and the scenario-wise subadditivity inequality fails.
\end{proposition}

\begin{proof}
The two-point loss \(X\) is reproduced exactly, so its optimized cost is
zero.  The calculation for \(Y\) is the second half of the proof of
\Cref{prop:monotonicity-fails}: its optimal partition is
\(\{0\}\mid\{2,3\}\), with levels \((0,5/2)\) and cost \(1/6\).
For \(X+Y\), splitting \(\{0,2\}\mid\{5\}\) gives levels \((1,5)\) and
cost
\[
 \frac{(0-1)^2+(2-1)^2}{3}=\frac23,
\]
whereas splitting \(\{0\}\mid\{2,5\}\) gives cost
\[
 \frac{(2-\tfrac72)^2+(5-\tfrac72)^2}{3}=\frac32.
\]
Hence the stated partition is unique.  At the third state the combined
report is \(5\), while the sum of the separate reports is \(2+5/2\);
at the first state the combined lower report \(1\) also exceeds \(0+0\).
\end{proof}

The minimized distortion is not a monetary risk measure either.  To state
the useful remnant in the correct units, for \(p\ge1\) define the
\(k\)-level \(L^p\) approximation error
\begin{equation}\label{eq:Dkp}
 D_{k,p}(X):=
 \inf_{z_1,\ldots,z_k\in\R}
 \left\|\min_{1\le i\le k}|X-z_i|\right\|_{L^p}.
\end{equation}
For \(G(u)=|u|^p\),
\(D_{2,p}(X)^p=Q_G(\Law(X))\).

\begin{proposition}[One-regime aggregation envelope]\label{prop:D-envelope}
For \(p\ge1\) and any jointly defined \(X,Y\in L^p\),
\begin{equation}\label{eq:D-envelope}
 D_{2,p}(X+Y)
 \le
 \min\!\left\{
  D_{2,p}(X)+D_{1,p}(Y),\
  D_{1,p}(X)+D_{2,p}(Y)
 \right\}.
\end{equation}
In general neither \(D_{2,p}\) nor \(Q_G\) is subadditive with its own
two-level value on the right-hand side.
\end{proposition}

\begin{proof}
For any two-point codebook \(A=\{z_1,z_2\}\) and \(m\in\R\),
\[
 \dist(X+Y,A+m)\le \dist(X,A)+|Y-m|
\]
pointwise.  Minkowski's inequality and then minimization over \(A\) and
\(m\) give the first bound inside the minimum; interchanging \(X,Y\) gives
the second.  In \Cref{prop:subadd-fails},
\[
 D_{2,2}(X)=0,\qquad
 D_{2,2}(Y)=\sqrt{\frac16},\qquad
 D_{2,2}(X+Y)=\sqrt{\frac23},
\]
which disproves two-level subadditivity.  The corresponding squared
distortions are \(0,1/6,2/3\), so \(Q_G\) fails it as well.
\end{proof}

\begin{remark}[What can and cannot be inherited from the block center]
For quadratic loss, the one-regime center is the mean and is monotone and
additive.  Upper expectiles generated by asymmetric quadratic loss satisfy
coherence properties for \(\tau\in[1/2,1)\), whereas quantiles
need not be subadditive
\citep{BelliniKlarMuellerRosazza2014,Artzner1999,RockafellarUryasev2002}.
The counterexamples
above use the mean inside every block.  They show that favorable axioms of
the one-regime center do not automatically pass through endogenous
partition optimization.  For essentially bounded losses, the support-hull
bound still gives the coarse control
\[
 b_{X+Y}\le\operatorname*{ess\,sup}(X+Y)
 \le\operatorname*{ess\,sup}X+\operatorname*{ess\,sup}Y
\]
for every active upper level \(b_{X+Y}\), but this is a worst-case envelope,
not subadditivity of the two-regime summary.
\end{remark}

\section{Stability of the two-regime summary}\label{sec:stability}

This section separates three questions.  Which topology on probability laws
controls an unbounded residual loss?  In which sense do cutoff profiles
converge when atoms create jumps?  Under what conditions do optimal cutoffs,
levels, and one-jump graph representatives converge?

\subsection{Why weak convergence is too weak for tail-sensitive centers}

\begin{example}[A vanishing tail with persistent loss]\label{ex:weak-fails}
Let \(G(u)=|u|\) and
\[
 \mu_n=\left(1-\frac1n\right)\delta_0+\frac1n\delta_n,
 \qquad \mu=\delta_0,\qquad n\ge2.
\]
Then \(\mu_n\Rightarrow\mu\).  For each fixed \(c<0\), however, the split has
only one active regime, and
\[
 H_G^\le(\mu_n,c)
 =\inf_m\int|x-m|\,\mu_n(dx)=1,
 \qquad
 H_G^\le(\mu,c)=0.
\]
Even \(\sup_n\int G((1+\eta)x)\,\mu_n(dx)<\infty\) for every fixed
\(\eta>0\).  The missing property is uniform integrability: the entire
first moment sits in a tail event whose probability tends to zero.
\end{example}

\subsection{The standard \texorpdfstring{\(\psi\)-weak}{psi-weak} topology}

Fix one \(\eta>0\) for which every law under comparison has the moment
specified below, and define
\begin{equation}\label{eq:psi}
 \psi(x):=1+G((1+\eta)x).
\end{equation}
By \Cref{ass:G}, \(\psi\) is continuous, at least one, and coercive.  Put
\[
\cP_\psi(\R):=
 \left\{\mu\in\cP(\R):\int\psi\,d\mu<\infty\right\}.
\]
Let \(\cM_\psi(\R)\) denote the corresponding finite nonnegative measures.

\begin{definition}[\(\psi\)-weak convergence]\label{def:psiweak}
For \(\nu_n,\nu\in\cM_\psi(\R)\), write \(\nu_n\psito\nu\) if
\begin{equation}\label{eq:psi-characterization}
 \nu_n\Rightarrow\nu,\qquad
 \int\psi\,d\nu_n\longrightarrow\int\psi\,d\nu.
\end{equation}
\end{definition}

\begin{lemma}[Useful characterizations]\label{lem:psi-characterization}
Condition \eqref{eq:psi-characterization} is equivalent to weak convergence
plus uniform \(\psi\)-integrability,
\[
 \lim_{M\to\infty}\sup_n
 \int_{\{\psi>M\}}\psi\,d\nu_n=0.
\]
It is also equivalent to convergence of integrals for every continuous
\(f\) satisfying \(|f|\le C\psi\).  These equivalences follow by truncating
at compact sublevel sets of the continuous coercive function \(\psi\).
\end{lemma}

\begin{proof}
Weak convergence controls \(\psi\wedge M\).  Together with convergence of
the \(\psi\)-moments, this makes the tails uniformly small; for example,
\[
 y\1_{\{y>2M\}}\le2(y-M)_+,\qquad
 \int(\psi-M)_+\,d\nu_n
 =\int\psi\,d\nu_n-\int(\psi\wedge M)\,d\nu_n.
\]
The finitely many initial indices are handled separately.  Conversely,
uniform tail control and weak convergence give moment convergence.
Truncating a continuous \(f\) with \(|f|\le C\psi\) on the compact sets
\(\{\psi\le M\}\) proves the test-function characterization.
\end{proof}

This is the standard weighted weak topology, used here with a gauge adapted
to \(G\); see
\citet{KratschmerSchiedZaehle2012,KratschmerSchiedZaehle2014}.

The choice \eqref{eq:psi} is dictated by the convex bound
\eqref{eq:convex-domination}: translated losses are controlled by
\(\psi\), uniformly when the center ranges over a compact set.

\begin{remark}[Wasserstein special cases]\label{rem:wasserstein}
For probability laws and \(G(u)=|u|^p\), \(p\ge1\), \(\psi\)-weak
convergence is exactly \(p\)-Wasserstein convergence: weak convergence plus
convergence of the \(p\)th absolute moments \citep{Villani2009}.  The same holds for losses
whose two tails are comparable to \(|u|^p\).  In particular, check,
absolute, and Huber losses induce the \(1\)-Wasserstein topology, whereas
quadratic and asymmetric quadratic losses induce the \(2\)-Wasserstein
topology.
\end{remark}

\subsection{Continuity of one-regime optimization}

\begin{lemma}[Continuity of \(G\)-centering]\label{lem:V-continuity}
Let \(\nu_n,\nu\in\cM_+(\R)\) have mass at most one and finite
\(\psi\)-moments.  If \(\nu_n\psito\nu\), then
\[
 \cV_G(\nu_n)\longrightarrow\cV_G(\nu).
\]
If \(\nu\ne0\), then \(\nu_n\ne0\) eventually; every sequence
\(m_n\in\cM_G(\nu_n)\) is then bounded, and every cluster point belongs to
\(\cM_G(\nu)\).  If the limiting center is unique, \(m_n\) converges to it.
\end{lemma}

\noindent See \hyperref[app:Vcontinuity]{\Cref*{app:Vcontinuity}} for the
proof.

\subsection{Moving cutoffs and profile regularity}

\begin{lemma}[Moving restrictions]\label{lem:moving-trunc}
Suppose \(\mu_n\psito\mu\), \(c_n\to c\in\R\), and \(\mu(\{c\})=0\).
Then
\[
 \restr{\mu_n}{(-\infty,c_n]}\psito
 \restr{\mu}{(-\infty,c]},\qquad
 \restr{\mu_n}{(c_n,\infty)}\psito
 \restr{\mu}{(c,\infty)}.
\]
The same statement holds for the open-left convention.
\end{lemma}

\noindent See \hyperref[app:movingtrunc]{\Cref*{app:movingtrunc}}.

\begin{theorem}[Continuous convergence at atom-free cutoffs]
\label{thm:moving-cut}
Suppose \(\mu_n\psito\mu\) and \(c_n\to c\in\Rbar\).
\begin{enumerate}[label=(\roman*),leftmargin=2.2em]
\item If \(c\in\R\) and \(\mu(\{c\})=0\), then
\[
 H_G^\le(\mu_n,c_n)\to
 H_G^\le(\mu,c)=H_G^<(\mu,c).
\]
The same limit holds for \(H_G^<\) and \(\underline H_G\).
\item If \(c=\pm\infty\), all three profiles converge to
\(\cV_G(\mu)\).
\item In case (i), suppose both limiting restrictions have positive mass
and a unique \(G\)-center.  Then every choice of optimal levels for the
moving restrictions converges to the corresponding limiting centers.
Strict convexity of \(G\) is a sufficient, but not necessary, condition.
\end{enumerate}
\end{theorem}

\begin{proof}
For finite atom-free \(c\), apply \Cref{lem:moving-trunc} to both
restrictions and then \Cref{lem:V-continuity}.  The center assertion is the
unique-center part of the same lemma.  At \(+\infty\), tightness and uniform
\(\psi\)-integrability make the upper restriction converge
\(\psi\)-weakly to zero and the lower restriction converge to \(\mu\); at
\(-\infty\) the roles reverse.  A detailed endpoint argument is included
in \hyperref[app:epi-endpoints]{the endpoint part of
\Cref*{app:epi}}.
\end{proof}

\begin{proposition}[Cutoff profiles at boundary atoms]\label{prop:profile-regularity}
For fixed \(\mu\), \(c\mapsto H_G^\le(\mu,c)\) is c\`adl\`ag on
\(\R\), with
\begin{equation}\label{eq:left-limit}
 H_G^\le(\mu,c-)=H_G^<(\mu,c).
\end{equation}
Its jumps can occur only at atoms of \(\mu\).  Moreover
\[
 \underline H_G(\mu,c)
 =\min\{H_G^\le(\mu,c-),H_G^\le(\mu,c)\}
\qquad(c\in\R).
\]
Together with the endpoint values
\(\underline H_G(\mu,\pm\infty)=H_G^\le(\mu,\pm\infty)=\cV_G(\mu)\),
this is the greatest lower-semicontinuous minorant of
\(H_G^\le(\mu,\cdot)\) on \(\Rbar\).  In particular,
\(\cC_G(\mu)\) is a nonempty compact subset of \(\Rbar\).
\end{proposition}

\noindent See
\hyperref[app:profile-regularity]{\Cref*{app:profile-regularity}}.

The lower-semicontinuous envelope is essential even in a three-point
example.  For \(G(u)=u^2\) and
\(\mu=(\delta_0+\delta_9+\delta_{10})/3\),
\[
 H_G^\le(\mu,9-)=\frac16,\qquad
 H_G^\le(\mu,9)=\frac{27}{2}.
\]
Thus the raw profile can jump upward at an atom; reversing the boundary
assignment gives the complementary one-sided value.  The relaxed profile
retains the better whole-atom assignment and is lower semicontinuous.

\begin{corollary}[Pointwise, local \(L^p\), and uniform profile limits]
\label{cor:profile-modes}
If \(\mu_n\psito\mu\), then:
\begin{enumerate}[label=(\alph*),leftmargin=2.2em]
\item \(H_G^\le(\mu_n,c)\to H_G^\le(\mu,c)\) at every atom-free \(c\);
\item for every compact \(K\subset\R\) and \(1\le p<\infty\),
\[
 \int_K|H_G^\le(\mu_n,c)-H_G^\le(\mu,c)|^p\,dc\to0;
\]
\item if \(\mu\) is atomless, then convergence is uniform on \(\Rbar\).
\end{enumerate}
\end{corollary}

\noindent See \hyperref[app:profile-modes]{\Cref*{app:profile-modes}}.

\subsection{Epi-convergence at atoms}\label{sec:epi}

Continuity-point convergence alone is not an optimizer theorem.  General
stability and epi-convergence theory for integrals and expectation
functionals under varying measures, including discontinuous integrands, is
available
\citep{ArtsteinWets1988,ArtsteinWets1994,VogelLachout2003,
FeinbergKasyanovRoyset2023}; standard variational consequences are developed
by \citet{RockafellarWets1998}.  The issue specific to the present profile is
the moving discontinuous threshold and its two possible whole-atom boundary
assignments, with the center in each moving restriction itself optimized.  At
a limiting atom, the lower-bound argument must also compare a possible
fractional limiting split with the two admissible endpoint assignments.  We
recall the sequential definition needed here.

\begin{definition}[Epi-convergence]\label{def:epi}
Functions \(f_n:\Rbar\to\R\cup\{+\infty\}\) epi-converge to \(f\) if:
\begin{enumerate}[label=(\roman*),leftmargin=2.2em]
\item for every \(c_n\to c\),
\(\liminf_n f_n(c_n)\ge f(c)\);
\item for every \(c\), there exists \(c_n\to c\) such that
\(\limsup_n f_n(c_n)\le f(c)\).
\end{enumerate}
\end{definition}

The first condition prevents spurious low valleys; the second supplies a
recovery sequence for every limiting point.

\begin{theorem}[Cutoff-profile epi-convergence with boundary atoms]\label{thm:epi}
If \(\mu_n\psito\mu\), then
\begin{equation}\label{eq:epi-conv}
 \underline H_G(\mu_n,\cdot)
 \xrightarrow{\mathrm{epi}}
 \underline H_G(\mu,\cdot)
 \qquad\text{on }\Rbar.
\end{equation}
\end{theorem}

\noindent The complete proof, including the atom and endpoint arguments, is
in \hyperref[app:epi]{\Cref*{app:epi}}.  Its key step is worth recording.
If \(p=\mu(\{c\})\), a moving sequence of discrete cuts can leave a fraction
\(\theta\in[0,1]\) of the limiting atom below the cut.  The limiting cost is
\[
 J_c(\theta)=
 \cV_G\!\left(\restr{\mu}{(-\infty,c)}+\theta p\delta_c\right)
 +\cV_G\!\left(\restr{\mu}{(c,\infty)}+(1-\theta)p\delta_c\right).
\]
The map \(\nu\mapsto\cV_G(\nu)\) is concave, being an infimum of linear
maps of \(\nu\).  Hence \(J_c\) is concave and
\[
 J_c(\theta)\ge \min\{J_c(0),J_c(1)\}
 =\underline H_G(\mu,c).
\]
Thus the two endpoint assignments, not a fractional split of a physical
atom, determine the stable variational envelope.

For sets \(A_n\subset\Rbar\), define the Painlev\'e--Kuratowski outer limit
\[
 \PKlimsup_{n\to\infty}A_n
 :=\{c:\text{some }n_k\uparrow\infty
       \text{ and }c_{n_k}\in A_{n_k}\text{ satisfy }c_{n_k}\to c\}.
\]

\begin{corollary}[Outer stability of exact and approximate cutoffs]
\label{cor:argmin}
Let
\[
 v_n:=\min_{c\in\Rbar}\underline H_G(\mu_n,c),\qquad
 v:=\min_{c\in\Rbar}\underline H_G(\mu,c).
\]
Under \(\mu_n\psito\mu\),
\[
 v_n\to v,\qquad
 \PKlimsup_{n\to\infty}\cC_G(\mu_n)\subseteq\cC_G(\mu).
\]
More generally, if \(\varepsilon_n\downarrow0\) and
\[
 \underline H_G(\mu_n,\widetilde c_n)\le v_n+\varepsilon_n,
\]
then every cluster point of \((\widetilde c_n)\) belongs to
\(\cC_G(\mu)\).  If \(\cC_G(\mu)=\{c^\star\}\), every exact or
\(\varepsilon_n\)-optimal selection converges to \(c^\star\).
\end{corollary}

\noindent See \hyperref[app:argmin]{\Cref*{app:argmin}}.

Here and below, unqualified \(L^p\) spaces for functions of the cutoff or
state variable use Lebesgue measure.

\subsection{Convergence of selected one-jump representatives}

At an atom-free limiting cutoff, the value assigned at that point has no
effect on the criterion.  To compare graphs in a standard path space, use
the right-continuous representative
\begin{equation}\label{eq:cadlag-regime}
 r_{a,b,c}(x):=a\1_{\{x<c\}}+b\1_{\{x\ge c\}}.
\end{equation}
On a compact interval \(K=[L,U]\), the Skorokhod \(J_1\) topology can be
described through increasing homeomorphisms \(\lambda:K\to K\): convergence
means that one can choose \(\lambda_n\) with
\[
 \|\lambda_n-\mathrm{id}\|_\infty\to0,\qquad
 \|r_n\circ\lambda_n-r\|_\infty\to0.
\]
Local \(J_1\) convergence means this on every compact interval whose
endpoints are continuity points of the limit.  This definition is
sufficient for the one-jump functions used below; see
\citet{Whitt2002} for the general path-space theory.

\begin{theorem}[Joint convergence of cutoffs, levels, and step graphs]
\label{thm:regime-conv}
Assume \(\mu_n\psito\mu\), and the limiting relaxed profile has a unique
finite minimizer \(c^\star\) satisfying
\[
 \mu(\{c^\star\})=0,\qquad
 0<\mu((-\infty,c^\star))<1.
\]
Assume also that the two limiting restrictions have unique \(G\)-centers
\(a^\star,b^\star\).  Let \(c_n\in\cC_G(\mu_n)\), and choose an action
\(\theta_n=(a_n,b_n,c_n,\sigma_n)\in\mathcal A_G(\mu_n)\).  Then
\[
 c_n\to c^\star,\qquad a_n\to a^\star,\qquad b_n\to b^\star.
\]
For the canonical right-continuous graph representatives
\(r_n=r_{a_n,b_n,c_n}\) and
\(r=r_{a^\star,b^\star,c^\star}\):
\begin{enumerate}[label=(\roman*),leftmargin=2.2em]
\item \(r_n(x)\to r(x)\) at every \(x\ne c^\star\);
\item \(r_n\to r\) in \(L^p_{\mathrm{loc}}(\R)\) for every finite
  \(p\ge1\);
\item \(r_n\to r\) in the local Skorokhod \(J_1\) topology.
\end{enumerate}
If the cost-realizing convention assigns an atom at \(c_n\) to the lower
cell, its exact optimal encoding need not equal \(r_n\) at that one point.
Nevertheless \(\mu_n(\{c_n\})\to0\), so the discrepancy vanishes in
\(\mu_n\)-probability and in expected loss.  Strict convexity of \(G\)
implies the required uniqueness of the limiting centers.
\end{theorem}

\noindent See \hyperref[app:regime-conv]{\Cref*{app:regime-conv}}.

\begin{proposition}[Equivalent descriptions of a genuine one-jump limit]
\label{prop:step-equivalence}
Let \(a<b\), let \(a_n\le b_n\), and let \(r_n,r\) have the form
\eqref{eq:cadlag-regime}.  Suppose the parameter sequences are bounded.
On every compact interval containing \(c\) in its interior, the following
are equivalent:
\begin{enumerate}[label=(\alph*),leftmargin=2.2em]
\item \((a_n,b_n,c_n)\to(a,b,c)\);
\item \(r_n(x)\to r(x)\) at every continuity point of \(r\);
\item \(r_n\to r\) in the \(J_1\) topology.
\end{enumerate}
Each condition implies Lebesgue \(L^p\) convergence for every finite \(p\).
If
\(a=b\), the cutoff is not identifiable and only convergence of the
constant function is meaningful.
\end{proposition}

\begin{proof}
(a) implies (b) directly.  For (a) implies (c), on
\([L,U]\) choose an increasing piecewise-linear homeomorphism fixing
\(L,U\) and satisfying \(\lambda_n(c)=c_n\).  Then
\(\|\lambda_n-\mathrm{id}\|_\infty\to0\), and
\[
 \sup_{x\in[L,U]}|r_n(\lambda_n(x))-r(x)|
 \le\max\{|a_n-a|,|b_n-b|\}\to0.
\]

Assume (b).  For every \(\varepsilon>0\) small enough that
\(c-\varepsilon,c+\varepsilon\in(L,U)\), the two limiting values at these
points are the distinct numbers \(a,b\).  For all sufficiently large \(n\),
the evaluations therefore lie on different plateaus, so
\[
 c-\varepsilon<c_n\le c+\varepsilon.
\]
Hence \(c_n\to c\).  Evaluating at any fixed points on the two sides then
gives \(a_n\to a\) and \(b_n\to b\), so (b) implies (a).  General \(J_1\)
convergence implies convergence at every
continuity point of the limit, hence (c) implies (b).  Finally, boundedness
and pointwise convergence away from one point give \(L^p\) convergence by
dominated convergence.
\end{proof}

\section{Discrete and empirical approximation}
\label{sec:approx}

The stability results turn finite-support calculations into consistent
plug-in approximations of the severity boundary and both conditional
centers for a general loss law.

Say that \(G\) has at most \(p\)-power growth if
\begin{equation}\label{eq:p-growth}
 G(u)\le K(1+|u|^p),\qquad u\in\R,
\end{equation}
for some \(K<\infty\) and \(p\ge1\).  The standard losses considered in
\Cref{sec:computation} satisfy this with \(p=1\) or \(p=2\).

\subsection{Deterministic finite-range approximation}

Let \(\pi_N:\R\to\R\) have finite range and put \(X^{(N)}=\pi_N(X)\).  A
concrete construction clips at radius \(R_N\) and rounds to a grid of mesh
\(\Delta_N\):
\begin{equation}\label{eq:quantizer}
 \pi_N(x)=\Delta_N\,\operatorname{round}\!\left(
 \frac{\max(-R_N,\min(x,R_N))}{\Delta_N}\right),
 \qquad R_N\uparrow\infty,\quad\Delta_N\downarrow0.
\end{equation}

\begin{theorem}[Consistency of finite-support approximation]
\label{thm:det-discrete}
Let \(\mu=\Law(X)\in\cP_\psi(\R)\), and suppose
\begin{equation}\label{eq:quantizer-ui}
 \pi_N(X)\to X\ \text{in probability},\qquad
 \{\psi(\pi_N(X)):N\ge1\}\ \text{is uniformly integrable}.
\end{equation}
Then, for \(\mu_N=\Law(\pi_N(X))\),
\[
 \mu_N\psito\mu.
\]
Consequently all conclusions of
\Cref{thm:moving-cut,cor:profile-modes,thm:epi,cor:argmin} apply to
\(\mu_N\).  Under the assumptions of \Cref{thm:regime-conv}, the optimal
discrete levels and cutoffs converge, as do their canonical right-continuous
graph representatives in the senses stated there.
\end{theorem}

\begin{proof}
Convergence in probability implies convergence in distribution, hence
\(\mu_N\Rightarrow\mu\).  Continuity of \(\psi\) gives
\(\psi(\pi_N(X))\to\psi(X)\) in probability.  Uniform integrability turns
this into convergence in \(L^1\), so
\[
 \int\psi\,d\mu_N=\E[\psi(\pi_N(X))]
 \longrightarrow\E[\psi(X)]=\int\psi\,d\mu.
\]
This is \(\psi\)-weak convergence.  The remaining claims are direct
applications of the cited stability results.
\end{proof}

If \(G\) satisfies \eqref{eq:p-growth}, convergence
\(\|\pi_N(X)-X\|_{L^p}\to0\) implies \eqref{eq:quantizer-ui}.  Indeed,
\(|\pi_N(X)|^p\to|X|^p\) in \(L^1\), and \(\psi(\pi_N(X))\) is dominated in
growth by \(1+|\pi_N(X)|^p\).  The clipping-and-rounding construction
\eqref{eq:quantizer} has this property whenever \(\E|X|^p<\infty\), since
\[
 |\pi_N(X)-X|\le (|X|-R_N)_++\frac{\Delta_N}{2}.
\]

For completeness, the \(p\)-Wasserstein distance is
\[
 \W_p(\mu,\nu):=
 \left(
 \inf_{\pi\in\Pi(\mu,\nu)}
 \int_{\R^2}|x-y|^p\,\pi(dx,dy)
 \right)^{1/p},
\]
where \(\Pi(\mu,\nu)\) denotes all couplings of \(\mu\) and \(\nu\).

\begin{corollary}[Wasserstein approximation]\label{cor:wass-discrete}
If \(G\) satisfies \eqref{eq:p-growth} and
\(\W_p(\mu_N,\mu)\to0\), then \(\mu_N\psito\mu\).  The profile and
optimal-value conclusions, outer optimizer-set stability, and the
unique-limit graph conclusions therefore apply.
\end{corollary}

\begin{proof}
\(\W_p\)-convergence is equivalent to weak convergence plus convergence
of \(p\)th moments.  It implies uniform \(p\)-moment integrability.  Since
\(\psi\) is continuous and has at most \(p\)-power growth, truncation on
compact sets followed by uniform tail control gives
\(\int\psi\,d\mu_N\to\int\psi\,d\mu\).  Apply
\Cref{def:psiweak}.
\end{proof}

\subsection{Empirical approximation}

Let \(X_1,X_2,\ldots\) be independent with common law \(\mu\), and define
\[
 \widehat\mu_n:=\frac1n\sum_{i=1}^n\delta_{X_i}.
\]

\begin{theorem}[Strong consistency of empirical regimes]
\label{thm:empirical}
If \(\E[\psi(X_1)]<\infty\), then
\[
 \widehat\mu_n\psito\mu
 \qquad\text{almost surely}.
\]
More precisely, there is one probability-one event on which
\[
 \min_c\underline H_G(\widehat\mu_n,c)
 \longrightarrow
 \min_c\underline H_G(\mu,c),\qquad
 \PKlimsup_{n\to\infty}\cC_G(\widehat\mu_n)
 \subseteq\cC_G(\mu).
\]
If the population cutoff is unique and the hypotheses of
\Cref{thm:regime-conv} hold, then on the same event every pathwise choice
of an exact empirical cutoff and active levels yields convergent cutoffs and
levels, together with convergent right-continuous graph representatives.
\end{theorem}

\begin{proof}
The empirical laws converge weakly almost surely.  The strong law of large
numbers gives, on the same probability-one event,
\[
 \int\psi\,d\widehat\mu_n
 =\frac1n\sum_{i=1}^n\psi(X_i)
 \longrightarrow\E[\psi(X_1)]
 =\int\psi\,d\mu.
\]
Thus \(\widehat\mu_n\psito\mu\) almost surely.  Apply
\Cref{thm:epi,cor:argmin,thm:regime-conv} pathwise.
\end{proof}

For quadratic loss this is an ordered two-cell counterpart of classical
strong consistency for \(k\)-means \citep{Pollard1981}.  The result also
covers nonsmooth and asymmetric residual losses and specifies the topology
needed to control the optimized tail loss.  General consistency results for
exact and approximate stochastic-program minimizers provide the broader
framework \citep{ArtsteinWets1995}; their application here uses the cutoff
profile at atoms.

\subsection{Approximation error beyond qualitative consistency}

\Cref{cor:argmin} allows the discrete profile itself to be minimized only
approximately: if the solver tolerance \(\varepsilon_N\) tends to zero,
every cluster point of an \(\varepsilon_N\)-optimal cutoff is
population-optimal.  Thus total error separates into

\begin{enumerate}[label=(\alph*),leftmargin=2.2em]
\item a law approximation error, controlled qualitatively by
  \(\psi\)-weak or Wasserstein convergence; and
\item a finite-support optimization error, certified by
  \Cref{prop:numerical-certificate} below.
\end{enumerate}

These results are consistency statements, not rates.  A quantitative rate
for a unique finite cutoff requires additional separation or margin
assumptions, such as a specified growth modulus
\[
 \underline H_G(\mu,c)-\underline H_G(\mu,c^\star)
 \ge\kappa\,|c-c^\star|^\alpha
\]
near \(c^\star\), together with a quantitative uniform bound for the
profile approximation.  Without such a condition, nearly tied split costs
can make the cutoff statistically unstable even when the minimum value is
well estimated.

\section{Finite-support computation and error control}\label{sec:computation}

Actuarial, reliability, and scenario models often provide a discrete loss
law directly.  With exact block costs, the exhaustive scan below returns all
optimal split intervals.  With approximate costs, it returns an approximate
candidate set and a conditional objective-error bound; completeness requires
the tolerance condition in \Cref{prop:numerical-certificate}.  The scan
coincides with the classical one-dimensional two-level scalar-quantizer
design \citep{Max1960,Lloyd1982,GrayNeuhoff1998}.  We formulate it as a
documented computational interface for loss-based risk summaries and connect
its error control to the stability theorem.

Let
\begin{equation}\label{eq:finite-law}
 \mu=\sum_{i=1}^m p_i\delta_{x_i},\qquad
 x_1<\cdots<x_m,\quad p_i>0,\quad\sum_{i=1}^m p_i=1,\quad m\ge1.
\end{equation}
Repeated support values must first be collapsed by adding their weights:
a measurable threshold cannot divide one atom internally.

For \(j=1,\ldots,m-1\), define
\begin{equation}\label{eq:Cj-general}
 C_j:=
 \cV_G\!\left(\sum_{i\le j}p_i\delta_{x_i}\right)
 +\cV_G\!\left(\sum_{i>j}p_i\delta_{x_i}\right),
\end{equation}
and let \(C_0:=\cV_G(\mu)\) denote the one-regime candidate.

\begin{proposition}[Finite-support prefix characterization]\label{prop:finite-exact}
For \(\mu\) in \eqref{eq:finite-law},
\[
 \min_{c\in\Rbar}\underline H_G(\mu,c)
 =\min_{0\le j\le m-1}C_j.
\]
If \(m=1\), the unique optimal encoding is the constant \(x_1\), and every
cutoff in \(\Rbar\) is optimal.  If \(m\ge2\) and \(j\ge1\) minimizes
\(C_j\), every cutoff in the closed gap
\([x_j,x_{j+1}]\) realizes that split under one of the two boundary
conventions.  Conversely, every nonconstant optimal threshold corresponds
to a minimizing split index.  Thus complete enumeration returns a global,
not merely local, optimum.  For \(m\ge2\), the one-regime cost \(C_0\) is
strictly larger than at least one nontrivial split cost.
\end{proposition}

\begin{proof}
For \(m=1\), every threshold has one active cell, its center is \(x_1\), and
the cost is zero.  Now let \(m\ge2\).  Every threshold partition of the
ordered distinct support is a prefix and its complementary suffix, except
for the empty/full one-regime endpoints.
For a fixed prefix, \Cref{prop:conditional-centers} makes the two block
optimizations independent, giving exactly \eqref{eq:Cj-general}.  A cutoff
strictly inside \((x_j,x_{j+1})\) realizes split \(j\) under either
convention.  At \(x_j\) or \(x_{j+1}\), the boundary relaxation chooses the
convention that still realizes split \(j\), hence the full closed gap.
No other partition is induced by a threshold.
Finally, let \(m_0\) be a center of the full law.  At least one extreme
support value, say \(x_e\in\{x_1,x_m\}\), differs from \(m_0\).  Isolate
that extreme in a singleton cell and evaluate the other cell at \(m_0\).
The resulting cost is at most
\[
 C_0-p_eG(x_e-m_0)<C_0,
\]
where the strict inequality uses the unique minimum of \(G\) at zero.
\end{proof}

\begin{algorithm}[Two-regime scan with conditional error control]
\label{alg:master}
\begin{enumerate}[leftmargin=2.2em]
\item \emph{Input.} Supply a finite support, convex loss \(G\), block-cost
error bound \(\varepsilon_{\rm block}\), and selection tolerance
\(\varepsilon_{\rm sel}\).  Use the uniform candidate-cost bound
\(\varepsilon_C:=2\varepsilon_{\rm block}\), since a nontrivial candidate
adds two block costs (the one-block candidate satisfies the same bound).
\item Sort values, reject invalid weights, and collapse exact duplicates.
If one support value remains, return its constant encoding and the full
cutoff set \(\Rbar\), then stop.
\item Compute a center and cost estimate \(\widehat C_0\) for the full
block.
\item For \(j=1,\ldots,m-1\), compute lower and upper block centers and
costs, and set \(\widehat C_j\) equal to their sum.
\item Set
\[
\widehat C_{\min}=\min_{0\le j\le m-1}\widehat C_j,\qquad
\widehat J=
\{j:\widehat C_j\le\widehat C_{\min}+\varepsilon_{\rm sel}\}.
\]
\item For every \(j\in\widehat J\setminus\{0\}\), report the candidate
interval \([x_j,x_{j+1}]\) and a display cutoff computed by an
overflow-safe midpoint routine, both regime levels, and the severe-regime
probability \(1-\sum_{i\le j}p_i\).
\item \emph{Output.} Report \(\widehat C_{\min}\), \(\widehat J\), the
candidate intervals, levels, severe-regime probabilities, and the bound from
\Cref{prop:numerical-certificate}.  Call \(\widehat J\) complete only when
its stated tolerance condition holds; exact arithmetic is the special case
\(\varepsilon_{\rm block}=\varepsilon_{\rm sel}=0\).
\end{enumerate}
\end{algorithm}

\subsection{Mean-based severity regimes and Otsu's criterion}

Let
\[
 P_j=\sum_{i\le j}p_i,\qquad
 S_j=\sum_{i\le j}p_ix_i,\qquad
 Q_j=\sum_{i\le j}p_ix_i^2.
\]
For \(G(u)=u^2\), the lower block has
\begin{equation}\label{eq:quad-left}
 a_j=\frac{S_j}{P_j},\qquad
 L_j=Q_j-\frac{S_j^2}{P_j},
\end{equation}
and the upper block has
\begin{equation}\label{eq:quad-right}
 b_j=\frac{S_m-S_j}{1-P_j},\qquad
 R_j=(Q_m-Q_j)-\frac{(S_m-S_j)^2}{1-P_j}.
\end{equation}
All costs are therefore obtained in \(O(m)\) arithmetic operations after
sorting.  The variance decomposition gives
\begin{equation}\label{eq:otsu}
 C_j=\Var_\mu(X)-P_j(1-P_j)(a_j-b_j)^2.
\end{equation}
Minimizing within-regime variance is exactly the two-class Otsu threshold
criterion \citep{Otsu1979}, and the same objective is the binary
Lloyd--Max distortion \citep{Max1960,Lloyd1982}.

The textbook expression \(Q-S^2/P\) can suffer cancellation when all
\(x_i\) have a large common offset but small spread.  In floating-point
implementations, compute moments after subtracting a fixed reference
value, or use a weighted Welford recurrence.  A tiny negative result caused
only by roundoff may be clamped to zero after checking it against a
scale-aware error bound.

\subsection{VaR-based severity regimes, stop-loss transforms, and absolute loss}

For \(0<\tau<1\), let
\[
 G_\tau(u)=\tau u_++(1-\tau)u_-.
\]
For a block \(l:r\), write \(W=P_r-P_{l-1}\), and choose the smallest
\(k\in\{l,\ldots,r\}\) satisfying
\begin{equation}\label{eq:quantile-index}
 P_k-P_{l-1}\ge\tau W.
\end{equation}
Then \(x_k\) is the smallest block \(\tau\)-quantile.  With
\[
\begin{array}{ll}
 W_-=P_{k-1}-P_{l-1},&S_-=S_{k-1}-S_{l-1},\\
 W_+=P_r-P_k,&S_+=S_r-S_k,
\end{array}
\]
the exact cost is
\begin{equation}\label{eq:check-cost}
 (1-\tau)(x_kW_--S_-)+\tau(S_+-x_kW_+).
\end{equation}
Binary search gives \(O(\log m)\) per block and \(O(m\log m)\) for all
splits.  In the prefix and suffix sweeps, the relevant weighted-quantile
indices move monotonically, so two maintained pointers reduce the scan to
\(O(m)\).  Absolute loss is twice the \(\tau=1/2\) check loss.

There is an exact actuarial identity.  If a nonzero block measure is
\(\nu=W\Law(Y)\), define the upper expected shortfall by
\[
 \operatorname{ES}_\tau(Y)
 =\inf_{m\in\R}
 \left\{m+\frac{1}{1-\tau}\E[(Y-m)_+]\right\}.
\]
Using
\((m-Y)_+=(Y-m)_++m-Y\) gives
\begin{equation}\label{eq:check-es}
 \cV_{G_\tau}(\nu)
 =W(1-\tau)
   \bigl(\operatorname{ES}_\tau(Y)-\E Y\bigr).
\end{equation}
Thus the block routine can equally be viewed as computing a scaled
expected-shortfall deviation through the stop-loss transform
\(m\mapsto\E[(Y-m)_+]\)
\citep{RockafellarUryasev2002,RockafellarUryasevZabarankin2006}.
For \(\tau=1/2\), the two-regime criterion is also the median-based
within-class absolute-deviation threshold studied by
\citet{XueTitterington2011}.  More generally, its fixed-\(\tau\)
counterpart is the one-dimensional, two-cluster specialization of
quantile-based clustering \citep{HennigViroliAnderlucci2019}.

\subsection{Expectile-based severity regimes}

For
\[
 G_\tau(u)=\tau u_+^2+(1-\tau)u_-^2,\qquad 0<\tau<1,
\]
the resulting objective is the ordered, fixed-\(\tau\), two-cluster
specialization of \(K\)-expectiles clustering
\citep{WangLiHaerdle2022}.  For one block, the center is the unique root of
\begin{equation}\label{eq:expectile-score}
 \tau\sum_{x_i>m}p_i(x_i-m)
 =(1-\tau)\sum_{x_i<m}p_i(m-x_i).
\end{equation}
If the root for block \(l:r\) lies in \([x_k,x_{k+1}]\), then
\begin{equation}\label{eq:expectile-piece}
 m=
 \frac{(1-\tau)(S_k-S_{l-1})+\tau(S_r-S_k)}
 {(1-\tau)(P_k-P_{l-1})+\tau(P_r-P_k)}.
\end{equation}
The score is continuous and strictly decreasing, so \(k\) is found by
binary search over support points.  Given \(k\), the cost is
\begin{align*}
 &(1-\tau)\!\left[
   (Q_k-Q_{l-1})-2m(S_k-S_{l-1})
   +m^2(P_k-P_{l-1})\right]\\
 &\quad+\tau\!\left[
   (Q_r-Q_k)-2m(S_r-S_k)+m^2(P_r-P_k)\right],
\end{align*}
which is \(O(1)\) from prefix sums.  Hence all split costs require
\(O(m\log m)\) time after sorting.  Monotone warm starts can improve this
in practice, but the stated bound does not rely on them.

\subsection{Huber loss}

For \(\delta>0\), define
\[
 G_\delta(u)=
 \begin{cases}
 \tfrac12u^2,&|u|\le\delta,\\
 \delta|u|-\tfrac12\delta^2,&|u|>\delta.
 \end{cases}
\]
The center set is the zero set of the continuous nonincreasing score
\begin{equation}\label{eq:huber-score}
 s_{l:r}(m)=
 \sum_{i=l}^r p_i\,\operatorname{clip}(x_i-m,-\delta,\delta).
\end{equation}
It may be an interval, for example for two balanced observations separated
by more than \(2\delta\).  We choose its smallest point for deterministic
output.  A one-sided bisection maintains a lower endpoint with
\(s_{l:r}(m)>0\) and an upper endpoint with \(s_{l:r}(m)\le0\); sending a
zero-score midpoint to the upper endpoint converges to the smallest zero
even when the score vanishes on an interval.  For a trial \(m\), binary
searches locate the indices below
\(m-\delta\), inside the quadratic band, and above \(m+\delta\).  Prefix
sums then evaluate both score and cost in \(O(\log m)\).

Bisection on \([x_l,x_r]\), after returning immediately when an endpoint is
a zero, maintains the stated strict/weak score bracket.  Reaching a center
bracket of width \(\varepsilon_m\) needs at most
\[
 N=
 \left\lceil
 \log_2\!\left(\max\left\{1,
 \frac{x_r-x_l}{\varepsilon_m}\right\}\right)
 \right\rceil,
\]
iterations.  Because the block objective is \(W\delta\)-Lipschitz in its
center, where \(W=\sum_{i=l}^rp_i\), this produces objective error at most
\(W\delta\varepsilon_m\).  The full scan therefore costs
\(O(m\log m\,N)\) with an explicit objective certificate.

\subsection{General smooth or nonsmooth convex loss}

For a differentiable \(G\), the finite-block score
\[
 s_{l:r}(m)=\sum_{i=l}^rp_iG'(x_i-m)
\]
is nonincreasing.  Bisection on \([x_l,x_r]\) is globally valid under
\Cref{ass:G}; a Newton step may be used only with a safeguard that preserves
the sign-changing bracket.  Direct score evaluation gives worst-case
\(O(m^2\log(1/\varepsilon_m))\) work for the full scan.

For a nonsmooth convex \(G\), compute
\[
 s^-_{l:r}(m)=\sum_{i=l}^rp_iG'_-(x_i-m),\qquad
 s^+_{l:r}(m)=\sum_{i=l}^rp_iG'_+(x_i-m).
\]
The exact center condition is
\begin{equation}\label{eq:nonsmooth-score}
 s^-_{l:r}(m)\le0\le s^+_{l:r}(m).
\end{equation}
The smallest satisfying \(m\) gives a deterministic endpoint of the center
interval.  Bisection should stop only when the center bracket or a
subgradient residual is small enough to imply the requested block
objective tolerance.

\begin{table}[t]
\centering
\caption{Arithmetic complexity after sorting and duplicate collapse.
Here \(k\) is a block size and \(N\) is a certified bisection count.}
\label{tab:complexity}
\begin{tabular}{lll}
\toprule
Loss & one block & all split costs\\
\midrule
quadratic & \(O(1)\) from \(P,S,Q\) & \(O(m)\)\\
check/absolute & \(O(\log m)\), or pointer update
  & \(O(m\log m)\), or \(O(m)\)\\
expectile & \(O(\log m)\) from \(P,S,Q\) & \(O(m\log m)\)\\
Huber & \(O(N\log m)\) & \(O(mN\log m)\)\\
general smooth convex & \(O(kN)\) & \(O(m^2N)\)\\
\bottomrule
\end{tabular}
\end{table}

\subsection{A posteriori numerical certification}

Root accuracy, objective accuracy, and tie tolerance are different units
and must not be denoted by one symbol.

\begin{proposition}[Objective-gap and minimizer-retention certificate]
\label{prop:numerical-certificate}
Let the exact candidates be \(C_0,\ldots,C_{m-1}\), and suppose
\[
 |\widehat C_j-C_j|\le\varepsilon_C
 \qquad(0\le j\le m-1).
\]
If an index \(\widehat j\) is selected with
\[
 \widehat C_{\widehat j}
 \le\min_j\widehat C_j+\varepsilon_{\rm sel},
\]
then
\begin{equation}\label{eq:global-certificate}
 C_{\widehat j}
 \le\min_jC_j+2\varepsilon_C+\varepsilon_{\rm sel}.
\end{equation}
If each of the two block costs has error at most
\(\varepsilon_{\rm block}\), one may take
\(\varepsilon_C=2\varepsilon_{\rm block}\).
Moreover, if \(\varepsilon_{\rm sel}\ge2\varepsilon_C\), then every exact
minimizer is retained:
\[
 \argmin_j C_j\subseteq\widehat J.
\]
Without this inequality, an exact minimizer can be omitted from
\(\widehat J\).
\end{proposition}

\begin{proof}
Let \(j^\star\) minimize the exact costs.  Then
\[
 C_{\widehat j}
 \le\widehat C_{\widehat j}+\varepsilon_C
 \le\widehat C_{j^\star}
   +\varepsilon_{\rm sel}+\varepsilon_C
 \le C_{j^\star}
   +\varepsilon_{\rm sel}+2\varepsilon_C.
\]
For completeness, let \(\widehat k\) minimize \(\widehat C_j\).  If
\(j^\star\) is any exact minimizer, then
\[
 \widehat C_{j^\star}-\widehat C_{\widehat k}
 \le (C_{j^\star}+\varepsilon_C)
      -(C_{\widehat k}-\varepsilon_C)
 \le2\varepsilon_C.
\]
Hence \(j^\star\in\widehat J\) under the stated condition.  Equality can
occur in both error bounds, so no smaller universal tolerance guarantees
inclusion.
\end{proof}

For a general convex loss on a block of diameter \(D\), let
\[
 L_D:=\max_{|u|\le D}
 \max\{|G'_-(u)|,|G'_+(u)|\}<\infty.
\]
The block objective is \(W L_D\)-Lipschitz on its support hull.  A center
bracket of width \(\varepsilon_m\) therefore contributes at most
\(WL_D\varepsilon_m\) to
\(\varepsilon_{\rm block}\).  With the conservative choice
\(\varepsilon_C=2\varepsilon_{\rm block}\), completeness of the approximate
candidate list is guaranteed by
\(\varepsilon_{\rm sel}\ge4\varepsilon_{\rm block}\).  Combining this
bound with
\eqref{eq:global-certificate} makes the phrase
``\(\varepsilon\)-optimal'' precise.  Floating-point comparisons can
certify only an approximate tie set unless interval arithmetic or exact
arithmetic separates the candidate-cost intervals.

If the law approximation index is \(N\), it is enough for
\(\varepsilon_{\rm sel,N}+2\varepsilon_{C,N}\to0\).  The approximate cuts
then satisfy the hypothesis of \Cref{cor:argmin}, so numerical inexactness
does not destroy consistency.

\section{A four-scenario loss model \texorpdfstring{\(\{0,3,u,v\}\)}
{with four severities}}\label{sec:fourpoint}

Consider four equally likely scenarios with loss severities
\(0,3,u,v\), represented by the law
\begin{equation}\label{eq:four-law}
 \mu_{u,v}=\frac14(\delta_0+\delta_3+\delta_u+\delta_v).
\end{equation}
Sorting removes the need to divide the \((u,v)\)-plane into all possible
orderings of \(0,3,u,v\).  Let
\[
 y_1\le y_2\le y_3\le y_4
\]
be the ordered multiset.  Under quadratic loss, the three possible
nontrivial prefix costs are
\begin{align}
 C_1&=\frac1{12}\left[
 (y_3-y_2)^2+(y_4-y_2)^2+(y_4-y_3)^2\right],
 \label{eq:C1-four}\\
 C_2&=\frac18\left[(y_2-y_1)^2+(y_4-y_3)^2\right],
 \label{eq:C2-four}\\
 C_3&=\frac1{12}\left[
 (y_2-y_1)^2+(y_3-y_1)^2+(y_3-y_2)^2\right].
 \label{eq:C3-four}
\end{align}

\begin{proposition}[Four-point specialization]\label{prop:four-rule}
For \eqref{eq:four-law} and \(G(u)=u^2\), let
\[
 J(u,v):=\{j\in\{1,2,3\}:y_j<y_{j+1}\},
\qquad
 J^\star(u,v):=\argmin_{j\in J(u,v)}C_j.
\]
For every \(j^\star\in J^\star(u,v)\), a displayed cutoff and the two
levels are
\begin{equation}\label{eq:four-cut}
 c^\star=\frac{y_{j^\star}+y_{j^\star+1}}2,
\end{equation}
and
\begin{equation}\label{eq:four-levels}
 a^\star=\frac1{j^\star}\sum_{i=1}^{j^\star}y_i,\qquad
 b^\star=\frac1{4-j^\star}\sum_{i=j^\star+1}^4y_i.
\end{equation}
The full optimal cutoff set is
\[
 \bigcup_{j\in J^\star(u,v)}[y_j,y_{j+1}].
\]
If support values coincide, indices with \(y_j=y_{j+1}\) are
inadmissible; equivalently, collapse the duplicate atom and add its
probabilities before running the weighted algorithm.
\end{proposition}

\begin{proof}
Every threshold-admissible lower regime is a prefix of the ordered
multiset.  A split through equal adjacent values would divide one atom and
is therefore inadmissible.  For \(k\) equally weighted observations, the
identity
\[
 \sum_{i=1}^k(z_i-\bar z)^2
 =\frac1k\sum_{1\le i<\ell\le k}(z_i-z_\ell)^2
\]
and the outer weight \(1/4\) give
\eqref{eq:C1-four}--\eqref{eq:C3-four}.  The conditional means give
\eqref{eq:four-levels}.  \Cref{prop:finite-exact} proves global optimality
and shows that a minimizing split contributes its full closed support gap.
The midpoint is only a deterministic display representative.  The fixed
values \(0\) and \(3\) make the law nondegenerate, so the one-regime
candidate is never optimal: isolating an extreme support atom at which a
one-regime center has nonzero residual strictly lowers its cost.  Thus the
three displayed split costs exhaust all optima.
\end{proof}

Inside any cone on which the ordering of \((0,3,u,v)\) is fixed, each
\(y_i\) is affine in \((u,v)\).  Pairwise tie boundaries
\(C_i=C_j\) are therefore algebraic curves of degree at most two.  Sorting
is shorter and less error-prone than deriving a separate objective in each
ordering cone.

\begin{figure}[t]
\centering
\includegraphics[width=\textwidth]{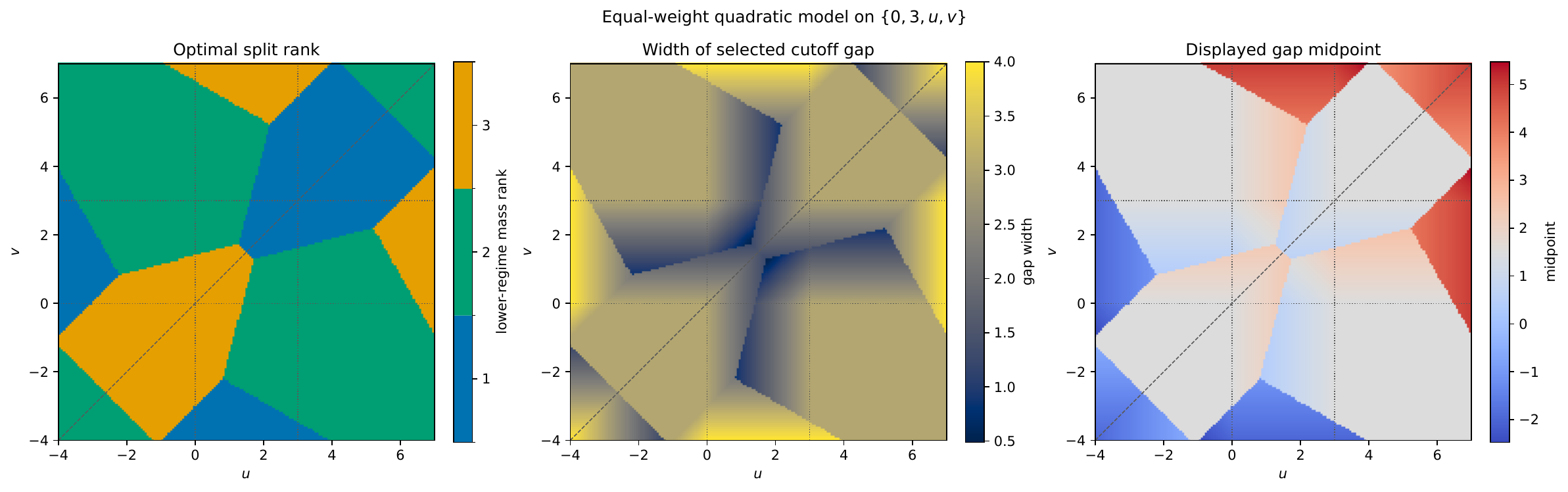}
\caption{Equal-weight quadratic four-point atlas on
\([-4,7]^2\), evaluated on a \(181\times181\) grid that contains the
collision coordinates \(0\) and \(3\) exactly.  Left: number of scenarios in
the displayed lower cell.  Middle: width of its full optimal cutoff gap.
Right: the midpoint used only for display.  On ties within numerical
tolerance, the image selects the
smallest split rank, while the mathematical summary contains the closed gap
for every minimizing rank.  Duplicate values are collapsed before
optimization.  The diagonal and the horizontal and vertical lines at \(0\)
and \(3\) are structural guides; cells are aligned with the sampled grid
nodes.  Other visible boundaries are raster approximations to quadratic tie
curves.}
\label{fig:four-atlas}
\end{figure}

\subsection{Interactive companion and reproducibility}

The self-contained file \texttt{code/four\_point\_explorer.html} provides:
\begin{enumerate}[label=(\roman*),leftmargin=2.2em]
\item sliders for \(u,v\) and a moving marker in the \((u,v)\)-plane;
\item the sorted support with masses, every admissible split, its full
  cutoff gap, candidate cost, and displayed midpoint;
\item the two center sets and a displayed representative associated with the
  deterministic tie rule;
\item quadratic, absolute, check, asymmetric quadratic, and Huber losses,
  with adjustable \(\tau\) and \(\delta\).
\end{enumerate}
The heatmaps are precomputed for quadratic loss; the detailed candidate
profile is recomputed for the chosen loss.  The Python generator fixes the
grid and writes both the static atlas and the HTML file.

The validation script compares every implemented one-block center against
an independent scalar optimizer, compares the closed four-point formulas
with direct block evaluation, and checks permutation symmetry, duplicate
collapse, score conditions, equivariance, invalid inputs, and tie cases.
Numerical comparisons report approximate tie sets at the stated absolute and
relative tolerances.  Exact algebraic certification would require symbolic or
interval-arithmetic bounds; the guarantee in
\Cref{prop:numerical-certificate} applies once certified block-cost bounds are
supplied.

\subsection{Unequal probabilities}

If \(0,3,u,v\) carry probabilities \(p_0,p_3,p_u,p_v\), sort
value--weight pairs together, collapse duplicates, and evaluate
\eqref{eq:Cj-general}.  No logic changes; only the compact equal-weight
formulas \eqref{eq:C1-four}--\eqref{eq:C3-four} are replaced by weighted
block costs.  Thus the four-point example is a transparent test case for the
general finite-support algorithm.

\section{Discussion}\label{sec:discussion}

The construction is best read as a compact severity model.  It replaces one
unconditional location by two conditional centers and lets the loss law
determine their boundary.  Here \emph{regime} denotes a static severity cell.
The choice of residual loss remains substantive: quadratic loss
emphasizes large errors, check loss yields conditional quantiles, asymmetric
quadratic loss yields conditional expectiles, and Huber loss limits the
influence of extremes.

This interpretation also marks the axiomatic boundary.  The correspondence is
law invariant, cash equivariant, and, for homogeneous losses, equivariant under
positive scaling.  Each fitted encoding is internally ordered in the realized
loss.  Yet reoptimizing the partition can reverse the fitted levels under an
ordered deterioration and can violate subadditivity, even for comonotone
losses.  Accordingly, the construction is used here as a descriptive risk
summary rather than as a coherent monetary risk measure or a set-valued
capital requirement of the usual upper-set kind \citep{HamelHeyde2010}.
Its optimizer multiplicity records tied partitions,
support gaps, and nonunique centers instead.  Likewise, the raw score in
\eqref{eq:joint-score} identifies a Bayes-action correspondence, but inactive
coordinates and boundary labels require the canonicalization and class-wide
uniqueness qualifications stated in \Cref{sec:risk-axioms}.

In an application, the summary should accompany the distribution rather than
replace it.  A displayed midpoint is only one point in an optimal support gap,
not a physical failure threshold or a regulatory limit.  An engineering study
should retain the hazard--fragility--loss chain and its exceedance curve; an
actuarial study should compare the fitted cells with VaR, expected shortfall,
expectiles, and the empirical loss distribution.  These comparisons provide
the context needed to assess practical decision performance.

By retaining the cutoff coordinate and its boundary convention, the stability
analysis complements classical codebook-consistency theory
\citep{AbayaWise1984,Pollard1981,Pollard1982,SabinGray1986,LiuPages2020,
Jaffe2025}.

Several limitations remain.  Optimizers can be nonunique, so reporting one
tie-broken output can hide genuine information.  Weak convergence alone does
not control an unbounded loss, and the required tail condition must be checked.
Epi-convergence gives consistency but no statistical rate or confidence set;
those need a separation condition for the population profile together with a
quantitative empirical or transportation bound.  Finally, two cells may be too
coarse.  With more ordered cells, finite-support dynamic programming is
available \citep{Fisher1958,WangSong2011,NielsenNock2014}, but the boundary-atom
convergence problem becomes multidimensional and is not proved here.

The stability results justify finite-support and empirical approximations when
the laws converge \(\psi\)-weakly: the relaxed profiles epi-converge and the
cutoff sets satisfy an outer-limit theorem.  Convergence of a selected cutoff
and its two levels additionally requires the uniqueness, atom-freeness, and
positive-cell assumptions of \Cref{thm:regime-conv}.  This separates stability
of the optimizer set from convergence of a particular reported summary in the
presence of ties and boundary atoms.

\section*{Data and code availability}
The numerical example is synthetic.  Source code, validation tests, the static
atlas, and a self-contained interactive explorer accompany the manuscript and
reproduce \Cref{fig:four-atlas}.  A public archival DOI should be added to the
accepted version.

\appendix
\section{Deferred proofs}\label{sec:appendix}

This appendix contains the measure-theoretic arguments deferred from the
main text.  The heading of each proof links to its statement; the final
line links back.

\subsection{Proof of \texorpdfstring{\Cref{prop:center-existence}}
{the center-existence proposition}}\label{app:center}

For \(m\in\R\), put
\[
 \Phi_\nu(m):=\int G(x-m)\,\nu(dx).
\]
\Cref{eq:convex-domination} makes \(\Phi_\nu(m)\) finite for every fixed
\(m\).  Convexity of \(G\) makes \(\Phi_\nu\) convex, and Fatou's lemma gives
lower semicontinuity.

Choose a compact interval \(K\) with \(\nu(K)>0\).  Coercivity gives
\[
 \inf_{x\in K}G(x-m)\longrightarrow\infty
 \qquad(|m|\to\infty),
\]
so
\[
 \Phi_\nu(m)\ge\nu(K)\inf_{x\in K}G(x-m)\longrightarrow\infty.
\]
Thus \(\Phi_\nu\) attains its minimum.  Its minimizer set is closed, bounded,
and convex, hence a nonempty compact interval.  If \(G\) is strictly
convex and \(m_1\ne m_2\), strict convexity holds pointwise for
\(G(x-m)\); integrating against the nonzero measure makes \(\Phi_\nu\)
strictly convex, so the minimizer is unique.

It remains to prove the support-hull assertion.  If
\(\alpha=\inf\supp\nu\) is finite and \(m<\alpha\), all residuals \(x-m\)
are positive.  A convex function with unique minimum at zero is strictly
increasing on \((0,\infty)\).  Moving \(m\) to the right by any
\(\varepsilon\in(0,\alpha-m]\) therefore strictly decreases
\(G(x-m)\) for every \(x\in\supp\nu\), and hence strictly decreases
\(\Phi_\nu\).  Such \(m\) cannot be a center.  The case
\(m>\sup\supp\nu\) is symmetric, using strict decrease of \(G\) on
\((-\infty,0)\).

\backlink{prop:center-existence}

\subsection{Proof of \texorpdfstring{\Cref{cor:conditional-order}}
{the conditional-order corollary}}\label{app:conditional-order}

The support of \(\mu_c^{\le,-}\) is contained in \((-\infty,c]\), and
the support of \(\mu_c^{\le,+}\) is contained in \([c,\infty)\).  The
support-hull part of \Cref{prop:center-existence} therefore gives
\[
 m_G^+(\mu_c^{\le,-})\le c\le
 m_G^-(\mu_c^{\le,+}).
\]

If \(c_1<c_2\), a direct comparison of distribution functions shows
\[
 \Law(X\mid X\le c_1)\preceq_{\mathrm{st}}
 \Law(X\mid X\le c_2),
\]
whenever both laws are defined.  Likewise
\[
 \Law(X\mid X>c_1)\preceq_{\mathrm{st}}
 \Law(X\mid X>c_2).
\]
\Cref{prop:center-order} proves monotonicity of both endpoints of both
conditional center sets.

As \(c\to+\infty\),
\(\restr{\mu}{(-\infty,c]}/\mu((-\infty,c])\psito\mu\); this follows
from monotone convergence of the restrictions and their \(\psi\)-moments.
The analogous upper conditional law converges to \(\mu\) as
\(c\to-\infty\).  If \(\mu\) has a unique center,
\Cref{lem:V-continuity} gives convergence of every conditional-center
selection.

\backlink{cor:conditional-order}

\subsection{Proof of \texorpdfstring{\Cref{lem:V-continuity}}
{continuity of one-regime centering}}\label{app:Vcontinuity}

For fixed \(m\), the function \(x\mapsto G(x-m)\) is continuous and,
by \eqref{eq:convex-domination}, bounded by \(C_m\psi(x)\).  Hence
\begin{equation}\label{eq:fixed-m-conv}
 \int G(x-m)\,\nu_n(dx)\longrightarrow
 \int G(x-m)\,\nu(dx).
\end{equation}

If \(\nu\ne0\), choose \(m\in\cM_G(\nu)\).  Then
\[
 \limsup_n\cV_G(\nu_n)
 \le\lim_n\int G(x-m)\,d\nu_n
 =\cV_G(\nu).
\]
If \(\nu=0\), convexity and \(G(0)=0\) give
\[
 G(x)\le\frac1{1+\eta}G((1+\eta)x).
\]
Choosing \(m=0\) and using \(\nu_n\psito0\) yields
\[
 0\le\cV_G(\nu_n)\le\int G(x)\,d\nu_n\longrightarrow0.
\]

Now assume \(\nu\ne0\), and take arbitrary
\(m_n\in\cM_G(\nu_n)\).  The preceding upper bound keeps the optimized
costs bounded.  Choose a compact interval \(K\) with
\(\nu(K)>0\) and \(\nu(\partial K)=0\).  Weak convergence makes
\(\nu_n(K)\) bounded away from zero.  If a subsequence had
\(|m_n|\to\infty\), then
\[
 \cV_G(\nu_n)
 \ge\nu_n(K)\inf_{x\in K}G(x-m_n)\longrightarrow\infty,
\]
a contradiction.  Thus the center sequence is bounded.

Let \((n_k)\) be an arbitrary subsequence.  It has a further subsequence,
not relabeled, with \(m_{n_k}\to m\).  Choose \(M\) containing these
centers.  Uniformly for \(|m'|\le M\),
\eqref{eq:convex-domination} gives
\[
 G(x-m')\le C_M\psi(x).
\]
On each compact \(x\)-interval,
\(G(x-m_{n_k})\to G(x-m)\) uniformly.  Outside a sufficiently large
compact set, uniform \(\psi\)-integrability makes the two translated-loss
tails uniformly small.  Combining the compact convergence, the tail
bound, and \eqref{eq:fixed-m-conv} gives
\[
 \int G(x-m_{n_k})\,d\nu_{n_k}
 \longrightarrow
 \int G(x-m)\,d\nu.
\]
Consequently
\[
 \cV_G(\nu)
 \le\int G(x-m)\,d\nu
 =\lim_k\cV_G(\nu_{n_k})
 \le\cV_G(\nu).
\]
Thus \(m\in\cM_G(\nu)\) and this further subsequence of values converges
to \(\cV_G(\nu)\).  Since every subsequence has such a further
subsequence, the full value sequence converges.  The same argument shows
that every center cluster point lies in \(\cM_G(\nu)\).  If that set is a
singleton, boundedness and uniqueness force convergence of the full center
sequence.

\backlink{lem:V-continuity}

\subsection{Proof of \texorpdfstring{\Cref{lem:moving-trunc}}
{the moving-restriction lemma}}\label{app:movingtrunc}

Let \(f\) be bounded and continuous.  Because \(\mu(\{c\})=0\), for every
\(\varepsilon>0\) one may choose \(\delta>0\) such that
\[
 \mu((c-2\delta,c+2\delta))<\varepsilon
\]
and \(c\pm2\delta\) are continuity points of \(\mu\).  For large \(n\),
\(c_n\in(c-\delta,c+\delta)\).  Continuous upper and lower
approximations to the half-line indicators agree outside
\((c-2\delta,c+2\delta)\).  Weak convergence controls their integrals,
whereas the remaining contribution is at most \(\|f\|_\infty\) times a
mass with limiting upper bound \(\varepsilon\).  Letting
\(\varepsilon\downarrow0\) proves weak convergence of both restrictions.

Apply the same bounded argument to \(\psi\wedge M\).  Uniform
\(\psi\)-integrability of the parent measures passes to every restriction:
\[
 \sup_n\int_{\{\psi>M\}}\psi\,
 d(\restr{\mu_n}{(-\infty,c_n]})
 \le
 \sup_n\int_{\{\psi>M\}}\psi\,d\mu_n
 \longrightarrow0.
\]
The complementary and open-left restrictions satisfy the same inequality.
Letting \(M\to\infty\) proves convergence of the \(\psi\)-moments.

\backlink{lem:moving-trunc}

\subsection{Proof of \texorpdfstring{\Cref{prop:profile-regularity}}
{profile regularity}}\label{app:profile-regularity}

If \(d_k\downarrow c\), continuity from above of the finite measure and
monotone convergence of its \(\psi\)-moments give
\[
 \restr{\mu}{(-\infty,d_k]}\psito
 \restr{\mu}{(-\infty,c]},\qquad
 \restr{\mu}{(d_k,\infty)}\psito
 \restr{\mu}{(c,\infty)}.
\]
\Cref{lem:V-continuity} gives right-continuity of \(H_G^\le\).  If
\(d_k\uparrow c\), the restrictions converge instead to those on
\((-\infty,c)\) and \([c,\infty)\), proving
\[
 H_G^\le(\mu,c-)=H_G^<(\mu,c).
\]
Hence the raw profile is c\`adl\`ag and can jump only where the two
boundary restrictions differ, namely at an atom.

For a c\`adl\`ag real function \(f\), the function
\(\widetilde f(c)=\min\{f(c-),f(c)\}\) is lower semicontinuous: from the
right, both nearby one-sided values tend to \(f(c)\), and from the left
they tend to \(f(c-)\); a general sequence is handled by splitting it into
left and right subsequences.  Thus \(\underline H_G\) is lower
semicontinuous at finite points.  At \(+\infty\), the upper restriction
vanishes in mass and \(\psi\)-moment while the lower restriction converges
to \(\mu\); the case \(-\infty\) is symmetric.  This proves lower
semicontinuity on compact \(\Rbar\).  A finite lower-semicontinuous
function on a compact space attains its minimum, and its argmin set is
compact.

It is also the greatest lower-semicontinuous minorant of the raw profile.
Indeed, if \(h\) is lower semicontinuous and
\(h\le H_G^\le(\mu,\cdot)\), then \(h(c)\le H_G^\le(\mu,c)\).  Along any
\(d_k\uparrow c\), lower semicontinuity and the preceding left-limit identity
also give
\[
 h(c)\le\liminf_k h(d_k)
 \le\liminf_k H_G^\le(\mu,d_k)
 =H_G^\le(\mu,c-).
\]
Thus \(h(c)\le\min\{H_G^\le(\mu,c),H_G^\le(\mu,c-)\}\); the endpoint cases
follow from the endpoint continuity just proved.

\backlink{prop:profile-regularity}

\subsection{Proof of \texorpdfstring{\Cref{cor:profile-modes}}
{the profile-convergence corollary}}\label{app:profile-modes}

Part (a) is \Cref{thm:moving-cut} with a fixed cutoff.  For every \(n,c\),
using one common center on both sides shows
\[
 0\le H_G^\le(\mu_n,c)\le\cV_G(\mu_n).
\]
\Cref{lem:V-continuity} makes the right-hand side uniformly bounded.
The atoms of \(\mu\) are countable, so part (a) gives almost-everywhere
convergence in \(c\).  Dominated convergence on each compact interval
proves part (b).

If \(\mu\) is atomless and uniform convergence on \(\Rbar\) failed, there
would be \(\varepsilon>0\), a subsequence, and \(c_n\in\Rbar\) such that
\[
 |H_G^\le(\mu_n,c_n)-H_G^\le(\mu,c_n)|\ge\varepsilon.
\]
Compactness gives a further subsequence with \(c_n\to c\).  By
\Cref{thm:moving-cut}, the first profile value converges to
\(H_G^\le(\mu,c)\).  The limiting profile is continuous because \(\mu\)
has no atoms, so the second does too, a contradiction.

\backlink{cor:profile-modes}

\subsection{Proof of \texorpdfstring{\Cref{thm:epi}}
{epi-convergence}}\label{app:epi}

\paragraph{Epi-liminf at a finite point.}
Let \(c_n\to c\in\R\), and pass to a subsequence that realizes the lower
limit of \(\underline H_G(\mu_n,c_n)\).  For each \(n\), choose whichever
boundary convention realizes the relaxed value, and denote the lower and
upper restricted measures by \(\nu_n^-,\nu_n^+\).  Then
\[
 0\le\nu_n^\pm\le\mu_n,\qquad
 \nu_n^-+\nu_n^+=\mu_n.
\]
Tightness and uniform \(\psi\)-integrability of \((\mu_n)\) pass to these
submeasures.  After simultaneous extraction,
\(\nu_n^\pm\Rightarrow\nu^\pm\).  Weak convergence of
\(\psi\wedge M\), followed by uniform tail control as \(M\to\infty\),
shows
\[
 \nu_n^\pm\psito\nu^\pm.
\]

Testing \(0\le\nu_n^-\le\mu_n\) against nonnegative continuous functions
and passing to the limit gives \(0\le\nu^-\le\mu\).  For every
\(f\in C_c((-\infty,c))\), the lower restriction eventually agrees with
\(\mu_n\); for every \(f\in C_c((c,\infty))\), it is eventually zero.
Therefore
\[
 \nu^-=\restr{\mu}{(-\infty,c)}+\theta\mu(\{c\})\delta_c
\]
for some \(\theta\in[0,1]\).  Since the two limits sum to \(\mu\),
\[
 \nu^+=\restr{\mu}{(c,\infty)}
 +(1-\theta)\mu(\{c\})\delta_c.
\]
\Cref{lem:V-continuity} gives, along this subsequence,
\[
 \underline H_G(\mu_n,c_n)
 \longrightarrow
 \cV_G(\nu^-)+\cV_G(\nu^+)=:J_c(\theta).
\]

For finite measures \(\alpha,\beta\) and \(t\in[0,1]\),
\begin{align*}
 \cV_G(t\alpha+(1-t)\beta)
 &=\inf_m\left[
 t\int G(x-m)\,d\alpha
 +(1-t)\int G(x-m)\,d\beta\right]\\
 &\ge t\cV_G(\alpha)+(1-t)\cV_G(\beta).
\end{align*}
Thus \(\cV_G\) is concave in its measure argument and \(J_c\) is concave
in \(\theta\).  Hence
\[
 J_c(\theta)\ge\min\{J_c(0),J_c(1)\}
 =\underline H_G(\mu,c),
\]
which is the epi-liminf inequality.

\paragraph{Recovery at a finite point.}
Suppose first that
\(H_G^\le(\mu,c)\le H_G^<(\mu,c)\).  Choose atom-free
\(d_k\downarrow c\).  Right-continuity gives
\[
 H_G^\le(\mu,d_k)\to H_G^\le(\mu,c).
\]
For fixed \(k\), \Cref{thm:moving-cut} gives
\[
 H_G^\le(\mu_n,d_k)\to H_G^\le(\mu,d_k).
\]
Choose the integers \(N_k\) strictly increasing so the absolute difference
is at most \(1/k\) for \(n\ge N_k\).  For \(n\ge N_1\), let \(k(n)\) be the
largest \(k\) with \(N_k\le n\), and define it arbitrarily for the finitely
many earlier indices.  Then \(k(n)\to\infty\), \(d_{k(n)}\to c\), and
\[
 \limsup_n\underline H_G(\mu_n,d_{k(n)})
 \le H_G^\le(\mu,c)
 =\underline H_G(\mu,c).
\]
If the open-left value is smaller, use atom-free \(d_k\uparrow c\) and
the left-limit identity.  Either construction works in a tie.

\paragraph{The two endpoints.}\label{app:epi-endpoints}
If \(c_n\to+\infty\), coercivity of \(\psi\) and uniform
\(\psi\)-integrability give
\[
 \int\psi\,d\nu_n^+\longrightarrow0
\]
for either boundary convention.  The lower restriction then converges
\(\psi\)-weakly to \(\mu\).  By \Cref{lem:V-continuity},
\[
 \underline H_G(\mu_n,c_n)\longrightarrow\cV_G(\mu).
\]
This proves both lower bound
and recovery at \(+\infty\); the argument at \(-\infty\) is symmetric.

\noindent\emph{\hyperref[thm:moving-cut]{Return to the endpoint statement
in \Cref*{thm:moving-cut}.}}

\backlink{thm:epi}

\subsection{Proof of \texorpdfstring{\Cref{cor:argmin}}
{optimizer stability}}\label{app:argmin}

Write \(f_n=\underline H_G(\mu_n,\cdot)\) and
\(f=\underline H_G(\mu,\cdot)\).  Let \(c^\star\) minimize \(f\).
A recovery sequence \(d_n\to c^\star\) gives
\[
 \limsup_n\min f_n\le\limsup_nf_n(d_n)\le\min f.
\]
Choose exact minimizers \(c_n\) of \(f_n\).  Compactness of \(\Rbar\)
gives a convergent subsequence \(c_{n_k}\to\bar c\), and the epi-liminf
inequality gives
\[
 \min f\le f(\bar c)\le
 \liminf_k f_{n_k}(c_{n_k})
 =\liminf_k\min f_{n_k}.
\]
Every subsequence has a further subsequence satisfying this lower bound,
so the minimum values converge.  The same display shows that every cluster
point of exact minimizers lies in \(\argmin f\).

If \(f_n(\widetilde c_n)\le\min f_n+\varepsilon_n\) with
\(\varepsilon_n\to0\), the identical compactness and epi-liminf argument
applies.  When \(\argmin f\) is a singleton, every cluster point is that
point; compactness then forces convergence of the full sequence.

\backlink{cor:argmin}

\subsection{Proof of \texorpdfstring{\Cref{thm:regime-conv}}
{selected-regime stability}}\label{app:regime-conv}

\Cref{cor:argmin} gives \(c_n\to c^\star\).  Let
\(\nu_n^-,\nu_n^+\) be the two restrictions under whichever boundary
convention realizes \(\underline H_G(\mu_n,c_n)\).  Since
\(\mu(\{c^\star\})=0\), \Cref{lem:moving-trunc} applies to either
convention and gives
\[
 \nu_n^-\psito\restr{\mu}{(-\infty,c^\star)},\qquad
 \nu_n^+\psito\restr{\mu}{(c^\star,\infty)}.
\]
Both limits are nonzero.  The assumed uniqueness of their centers and
\Cref{lem:V-continuity} therefore imply
\[
 a_n\to a^\star,\qquad b_n\to b^\star.
\]

Let \(\widetilde r_n=q_{\theta_n}\) denote the exact selected encoding under
the cost-realizing convention.  The right-continuous \(r_n\) in the theorem
is its canonical c\`adl\`ag graph representative; it need not itself be an
optimizer.  The two functions differ only when the convention puts an atom
at \(c_n\) in the lower regime.  The moving-atom mass vanishes:
for each \(\delta>0\), eventually
\[
 \mu_n(\{c_n\})
 \le\mu_n([c^\star-\delta,c^\star+\delta]),
\]
whose limiting upper bound tends to \(\mu(\{c^\star\})=0\) as
\(\delta\downarrow0\).  Thus
\[
 \mu_n(\widetilde r_n\ne r_n)
 \le \mu_n(\{c_n\})\longrightarrow0.
\]
The convention discrepancy is also asymptotically irrelevant for the
objective.  Indeed, \(a_n,b_n,c_n\) eventually lie in a
common compact interval, so continuity of \(G\) gives
\[
 \mu_n(\{c_n\})
 \left|G(c_n-a_n)-G(c_n-b_n)\right|\longrightarrow0.
\]
Midpoint cuts in support gaps have no discrepancy at all.

Parameter convergence gives pointwise convergence at every
\(x\ne c^\star\).  On a compact interval, the parameters and functions
are uniformly bounded, so dominated convergence gives every finite local
Lebesgue-\(L^p\) limit.  Finally, on a compact interval containing
\(c^\star\) in
its interior, choose the piecewise-linear increasing homeomorphism
\(\lambda_n\) that fixes the interval endpoints and maps \(c^\star\) to
\(c_n\).  Then
\[
 \|\lambda_n-\mathrm{id}\|_\infty\to0,\qquad
 \|r_n\circ\lambda_n-r\|_\infty
 \le\max\{|a_n-a^\star|,|b_n-b^\star|\}\to0.
\]
On a compact interval away from \(c^\star\), eventually both functions use
the same plateau and the identity time change suffices.  This is local
\(J_1\) convergence.

\backlink{thm:regime-conv}

\begingroup
\small
\bibliographystyle{plainnat}
\bibliography{references}

@article{AbayaWise1984,
  author  = {Abaya, Efren F. and Wise, Gary L.},
  title   = {Convergence of Vector Quantizers with Applications to Optimal Quantization},
  journal = {SIAM Journal on Applied Mathematics},
  year    = {1984},
  volume  = {44},
  number  = {1},
  pages   = {183--189},
  doi     = {10.1137/0144013}
}

@article{Artzner1999,
  author  = {Artzner, Philippe and Delbaen, Freddy and Eber, Jean-Marc and Heath, David},
  title   = {Coherent Measures of Risk},
  journal = {Mathematical Finance},
  year    = {1999},
  volume  = {9},
  number  = {3},
  pages   = {203--228},
  doi     = {10.1111/1467-9965.00068}
}

@article{ArtsteinWets1988,
  author  = {Artstein, Zvi and Wets, Roger J.-B.},
  title   = {Approximating the Integral of a Multifunction},
  journal = {Journal of Multivariate Analysis},
  year    = {1988},
  volume  = {24},
  number  = {2},
  pages   = {285--308},
  doi     = {10.1016/0047-259X(88)90041-3}
}

@article{ArtsteinWets1994,
  author  = {Artstein, Zvi and Wets, Roger J.-B.},
  title   = {Stability Results for Stochastic Programs and Sensors,
             Allowing for Discontinuous Objective Functions},
  journal = {SIAM Journal on Optimization},
  year    = {1994},
  volume  = {4},
  number  = {3},
  pages   = {537--550},
  doi     = {10.1137/0804030}
}

@article{BanerjeeMcKeague2007,
  author  = {Banerjee, Moulinath and McKeague, Ian W.},
  title   = {Confidence Sets for Split Points in Decision Trees},
  journal = {The Annals of Statistics},
  year    = {2007},
  volume  = {35},
  number  = {2},
  pages   = {543--574},
  doi     = {10.1214/009053606000001415}
}

@article{BuehlmannYu2002,
  author  = {B{\"u}hlmann, Peter and Yu, Bin},
  title   = {Analyzing Bagging},
  journal = {The Annals of Statistics},
  year    = {2002},
  volume  = {30},
  number  = {4},
  pages   = {927--961},
  doi     = {10.1214/aos/1031689014}
}

@article{ArtsteinWets1995,
  author  = {Artstein, Zvi and Wets, Roger J.-B.},
  title   = {Consistency of Minimizers and the {SLLN} for Stochastic Programs},
  journal = {Journal of Convex Analysis},
  year    = {1995},
  volume  = {2},
  number  = {1--2},
  pages   = {1--17},
  url     = {https://journalofconvexanalysis.com/articles/jca02001/}
}

@article{BelliniKlarMuellerRosazza2014,
  author  = {Bellini, Fabio and Klar, Bernhard and M{\"u}ller, Alfred and Rosazza Gianin, Emanuela},
  title   = {Generalized Quantiles as Risk Measures},
  journal = {Insurance: Mathematics and Economics},
  year    = {2014},
  volume  = {54},
  pages   = {41--48},
  doi     = {10.1016/j.insmatheco.2013.10.015}
}

@article{Fisher1958,
  author  = {Fisher, Walter D.},
  title   = {On Grouping for Maximum Homogeneity},
  journal = {Journal of the American Statistical Association},
  year    = {1958},
  volume  = {53},
  number  = {284},
  pages   = {789--798},
  doi     = {10.1080/01621459.1958.10501479}
}

@article{FaugerasPages2024,
  author  = {Faugeras, Olivier P. and Pag{\`e}s, Gilles},
  title   = {Risk Quantization by Magnitude and Propensity},
  journal = {Insurance: Mathematics and Economics},
  year    = {2024},
  volume  = {116},
  pages   = {134--147},
  doi     = {10.1016/j.insmatheco.2024.02.005}
}

@article{FeinbergKasyanovRoyset2023,
  author  = {Feinberg, Eugene A. and Kasyanov, Pavlo O. and Royset, Johannes O.},
  title   = {Epi-Convergence of Expectation Functions under Varying Measures and Integrands},
  journal = {Journal of Convex Analysis},
  year    = {2023},
  volume  = {30},
  number  = {3},
  pages   = {917--936},
  url     = {https://journalofconvexanalysis.com/articles/jca30043/}
}

@article{FoellmerSchied2002,
  author  = {F{\"o}llmer, Hans and Schied, Alexander},
  title   = {Convex Measures of Risk and Trading Constraints},
  journal = {Finance and Stochastics},
  year    = {2002},
  volume  = {6},
  number  = {4},
  pages   = {429--447},
  doi     = {10.1007/s007800200072}
}

@article{Gneiting2011,
  author  = {Gneiting, Tilmann},
  title   = {Making and Evaluating Point Forecasts},
  journal = {Journal of the American Statistical Association},
  year    = {2011},
  volume  = {106},
  number  = {494},
  pages   = {746--762},
  doi     = {10.1198/jasa.2011.r10138}
}

@article{GrayNeuhoff1998,
  author  = {Gray, Robert M. and Neuhoff, David L.},
  title   = {Quantization},
  journal = {IEEE Transactions on Information Theory},
  year    = {1998},
  volume  = {44},
  number  = {6},
  pages   = {2325--2383},
  doi     = {10.1109/18.720541}
}

@article{Huber1964,
  author  = {Huber, Peter J.},
  title   = {Robust Estimation of a Location Parameter},
  journal = {The Annals of Mathematical Statistics},
  year    = {1964},
  volume  = {35},
  number  = {1},
  pages   = {73--101},
  doi     = {10.1214/aoms/1177703732}
}

@article{KoenkerBassett1978,
  author  = {Koenker, Roger and Bassett, Jr., Gilbert},
  title   = {Regression Quantiles},
  journal = {Econometrica},
  year    = {1978},
  volume  = {46},
  number  = {1},
  pages   = {33--50},
  doi     = {10.2307/1913643}
}

@article{Jaffe2025,
  author  = {Jaffe, Adam Quinn},
  title   = {Asymptotic Theory of Geometric and Adaptive {$k$}-Means Clustering},
  journal = {The Annals of Statistics},
  year    = {2025},
  volume  = {53},
  number  = {4},
  pages   = {1559--1586},
  doi     = {10.1214/25-AOS2514}
}

@article{KratschmerSchiedZaehle2014,
  author  = {Kr{\"a}tschmer, Volker and Schied, Alexander and Z{\"a}hle, Henryk},
  title   = {Comparative and Qualitative Robustness for Law-Invariant Risk Measures},
  journal = {Finance and Stochastics},
  year    = {2014},
  volume  = {18},
  number  = {2},
  pages   = {271--295},
  doi     = {10.1007/s00780-013-0225-4}
}

@incollection{Kusuoka2001,
  author    = {Kusuoka, Shigeo},
  title     = {On Law Invariant Coherent Risk Measures},
  booktitle = {Advances in Mathematical Economics},
  editor    = {Kusuoka, Shigeo and Maruyama, Toru},
  publisher = {Springer},
  address   = {Tokyo},
  year      = {2001},
  volume    = {3},
  pages     = {83--95},
  doi       = {10.1007/978-4-431-67891-5_4}
}

@article{LiuPages2020,
  author  = {Liu, Yating and Pag{\`e}s, Gilles},
  title   = {Convergence Rate of Optimal Quantization and Application to the Clustering Performance of the Empirical Measure},
  journal = {Journal of Machine Learning Research},
  year    = {2020},
  volume  = {21},
  number  = {86},
  pages   = {1--36},
  url     = {https://jmlr.org/papers/v21/18-804.html}
}

@article{Lloyd1982,
  author  = {Lloyd, Stuart P.},
  title   = {Least Squares Quantization in {PCM}},
  journal = {IEEE Transactions on Information Theory},
  year    = {1982},
  volume  = {28},
  number  = {2},
  pages   = {129--137},
  doi     = {10.1109/TIT.1982.1056489}
}

@article{Max1960,
  author  = {Max, Joel},
  title   = {Quantizing for Minimum Distortion},
  journal = {IRE Transactions on Information Theory},
  year    = {1960},
  volume  = {6},
  number  = {1},
  pages   = {7--12},
  doi     = {10.1109/TIT.1960.1057548}
}

@article{NeweyPowell1987,
  author  = {Newey, Whitney K. and Powell, James L.},
  title   = {Asymmetric Least Squares Estimation and Testing},
  journal = {Econometrica},
  year    = {1987},
  volume  = {55},
  number  = {4},
  pages   = {819--847},
  doi     = {10.2307/1911031}
}

@article{NielsenNock2014,
  author  = {Nielsen, Frank and Nock, Richard},
  title   = {Optimal Interval Clustering: Application to {Bregman} Clustering and Statistical Mixture Learning},
  journal = {IEEE Signal Processing Letters},
  year    = {2014},
  volume  = {21},
  number  = {10},
  pages   = {1289--1292},
  doi     = {10.1109/LSP.2014.2333001}
}

@article{Otsu1979,
  author  = {Otsu, Nobuyuki},
  title   = {A Threshold Selection Method from Gray-Level Histograms},
  journal = {IEEE Transactions on Systems, Man, and Cybernetics},
  year    = {1979},
  volume  = {9},
  number  = {1},
  pages   = {62--66},
  doi     = {10.1109/TSMC.1979.4310076}
}

@article{Pollard1981,
  author  = {Pollard, David},
  title   = {Strong Consistency of $K$-Means Clustering},
  journal = {The Annals of Statistics},
  year    = {1981},
  volume  = {9},
  number  = {1},
  pages   = {135--140},
  doi     = {10.1214/aos/1176345339}
}

@article{Pollard1982,
  author  = {Pollard, David},
  title   = {Quantization and the Method of $k$-Means},
  journal = {IEEE Transactions on Information Theory},
  year    = {1982},
  volume  = {28},
  number  = {2},
  pages   = {199--205},
  doi     = {10.1109/TIT.1982.1056481}
}

@article{RockafellarUryasev2002,
  author  = {Rockafellar, R. Tyrrell and Uryasev, Stanislav},
  title   = {Conditional Value-at-Risk for General Loss Distributions},
  journal = {Journal of Banking \& Finance},
  year    = {2002},
  volume  = {26},
  number  = {7},
  pages   = {1443--1471},
  doi     = {10.1016/S0378-4266(02)00271-6}
}

@article{RockafellarUryasevZabarankin2006,
  author  = {Rockafellar, R. Tyrrell and Uryasev, Stanislav and Zabarankin, Michael},
  title   = {Generalized Deviations in Risk Analysis},
  journal = {Finance and Stochastics},
  year    = {2006},
  volume  = {10},
  number  = {1},
  pages   = {51--74},
  doi     = {10.1007/s00780-005-0165-8}
}

@book{RockafellarWets1998,
  author    = {Rockafellar, R. Tyrrell and Wets, Roger J.-B.},
  title     = {Variational Analysis},
  publisher = {Springer},
  address   = {Berlin},
  year      = {1998},
  series    = {Grundlehren der mathematischen Wissenschaften},
  volume    = {317},
  doi       = {10.1007/978-3-642-02431-3}
}

@article{SabinGray1986,
  author  = {Sabin, Michael J. and Gray, Robert M.},
  title   = {Global Convergence and Empirical Consistency of the Generalized {Lloyd} Algorithm},
  journal = {IEEE Transactions on Information Theory},
  year    = {1986},
  volume  = {32},
  number  = {2},
  pages   = {148--155},
  doi     = {10.1109/TIT.1986.1057168}
}

@book{Villani2009,
  author    = {Villani, C{\'e}dric},
  title     = {Optimal Transport: Old and New},
  publisher = {Springer},
  address   = {Berlin},
  year      = {2009},
  series    = {Grundlehren der mathematischen Wissenschaften},
  volume    = {338},
  doi       = {10.1007/978-3-540-71050-9}
}

@article{WangSong2011,
  author  = {Wang, Haizhou and Song, Mingzhou},
  title   = {{Ckmeans.1d.dp}: Optimal $k$-Means Clustering in One Dimension by Dynamic Programming},
  journal = {The R Journal},
  year    = {2011},
  volume  = {3},
  number  = {2},
  pages   = {29--33},
  doi     = {10.32614/RJ-2011-015}
}

@book{Whitt2002,
  author    = {Whitt, Ward},
  title     = {Stochastic-Process Limits: An Introduction to Stochastic-Process Limits and Their Application to Queues},
  publisher = {Springer},
  address   = {New York},
  year      = {2002},
  series    = {Springer Series in Operations Research and Financial Engineering},
  doi       = {10.1007/b97479}
}

@article{XueTitterington2011,
  author  = {Xue, Jing-Hao and Titterington, D. Michael},
  title   = {Median-Based Image Thresholding},
  journal = {Image and Vision Computing},
  year    = {2011},
  volume  = {29},
  number  = {9},
  pages   = {631--637},
  doi     = {10.1016/j.imavis.2011.06.003}
}

@book{GrafLuschgy2000,
  author    = {Graf, Siegfried and Luschgy, Harald},
  title     = {Foundations of Quantization for Probability Distributions},
  publisher = {Springer},
  address   = {Berlin},
  year      = {2000},
  series    = {Lecture Notes in Mathematics},
  volume    = {1730},
  doi       = {10.1007/BFb0103945}
}

@article{HennigViroliAnderlucci2019,
  author  = {Hennig, Christian and Viroli, Cinzia and Anderlucci, Laura},
  title   = {Quantile-Based Clustering},
  journal = {Electronic Journal of Statistics},
  year    = {2019},
  volume  = {13},
  number  = {2},
  pages   = {4849--4883},
  doi     = {10.1214/19-EJS1640}
}

@article{KratschmerSchiedZaehle2012,
  author  = {Kr{\"a}tschmer, Volker and Schied, Alexander and Z{\"a}hle, Henryk},
  title   = {Qualitative and Infinitesimal Robustness of Tail-Dependent Statistical Functionals},
  journal = {Journal of Multivariate Analysis},
  year    = {2012},
  volume  = {103},
  number  = {1},
  pages   = {35--47},
  doi     = {10.1016/j.jmva.2011.06.005}
}

@article{RockafellarUryasev2013,
  author  = {Rockafellar, R. Tyrrell and Uryasev, Stanislav},
  title   = {The Fundamental Risk Quadrangle in Risk Management, Optimization and Statistical Estimation},
  journal = {Surveys in Operations Research and Management Science},
  year    = {2013},
  volume  = {18},
  number  = {1--2},
  pages   = {33--53},
  doi     = {10.1016/j.sorms.2013.03.001}
}

@article{WangLiHaerdle2022,
  author  = {Wang, Bingling and Li, Yingxing and H{\"a}rdle, Wolfgang Karl},
  title   = {{K}-Expectiles Clustering},
  journal = {Journal of Multivariate Analysis},
  year    = {2022},
  volume  = {189},
  pages   = {104869},
  doi     = {10.1016/j.jmva.2021.104869}
}

@article{CascosMolchanov2007,
  author  = {Cascos, Ignacio and Molchanov, Ilya},
  title   = {Multivariate Risks and Depth-Trimmed Regions},
  journal = {Finance and Stochastics},
  year    = {2007},
  volume  = {11},
  number  = {3},
  pages   = {373--397},
  doi     = {10.1007/s00780-007-0043-7}
}

@article{JouiniMeddebTouzi2004,
  author  = {Jouini, Ely{\`e}s and Meddeb, Moncef and Touzi, Nizar},
  title   = {Vector-Valued Coherent Risk Measures},
  journal = {Finance and Stochastics},
  year    = {2004},
  volume  = {8},
  number  = {4},
  pages   = {531--552},
  doi     = {10.1007/s00780-004-0127-6}
}

@article{HamelHeyde2010,
  author  = {Hamel, Andreas H. and Heyde, Frank},
  title   = {Duality for Set-Valued Measures of Risk},
  journal = {SIAM Journal on Financial Mathematics},
  year    = {2010},
  volume  = {1},
  number  = {1},
  pages   = {66--95},
  doi     = {10.1137/080743494}
}

@article{FisslerZiegel2016,
  author  = {Fissler, Tobias and Ziegel, Johanna F.},
  title   = {Higher Order Elicitability and {Osband's} Principle},
  journal = {The Annals of Statistics},
  year    = {2016},
  volume  = {44},
  number  = {4},
  pages   = {1680--1707},
  doi     = {10.1214/16-AOS1439}
}

@article{FisslerFrongilloHlavinovaRudloff2021,
  author  = {Fissler, Tobias and Frongillo, Rafael and
             Hlavinov{\'a}, Jana and Rudloff, Birgit},
  title   = {Forecast Evaluation of Quantiles, Prediction Intervals, and
             Other Set-Valued Functionals},
  journal = {Electronic Journal of Statistics},
  year    = {2021},
  volume  = {15},
  number  = {1},
  pages   = {1034--1084},
  doi     = {10.1214/21-EJS1808}
}

@article{Eubank1988,
  author  = {Eubank, R. L.},
  title   = {Optimal Grouping, Spacing, Stratification, and Piecewise
             Constant Approximation},
  journal = {SIAM Review},
  year    = {1988},
  volume  = {30},
  number  = {3},
  pages   = {404--420},
  doi     = {10.1137/1030092}
}

@article{Flury1990,
  author  = {Flury, Bernard A.},
  title   = {Principal Points},
  journal = {Biometrika},
  year    = {1990},
  volume  = {77},
  number  = {1},
  pages   = {33--41},
  doi     = {10.1093/biomet/77.1.33}
}

@article{LiFlury1995,
  author  = {Li, Luning and Flury, Bernard},
  title   = {Uniqueness of Principal Points for Univariate Distributions},
  journal = {Statistics \& Probability Letters},
  year    = {1995},
  volume  = {25},
  number  = {4},
  pages   = {323--327},
  doi     = {10.1016/0167-7152(94)00237-3}
}

@article{MeaseNair2006,
  author  = {Mease, David and Nair, Vijayan N.},
  title   = {Unique Optimal Partitions of Distributions and Connections to
             Hazard Rates and Stochastic Ordering},
  journal = {Statistica Sinica},
  year    = {2006},
  volume  = {16},
  number  = {4},
  pages   = {1299--1312},
  url     = {https://www3.stat.sinica.edu.tw/statistica/j16n4/j16n412/j16n412.html}
}

@misc{NistorPopescu2026,
  author        = {Nistor, Mihaela-Adriana and Popescu, Ionel},
  title         = {A Property of Log-Concave and Weakly-Symmetric
                   Distributions for Two Step Approximations of Random Variables},
  year          = {2026},
  eprint        = {2603.12767},
  archivePrefix = {arXiv},
  primaryClass  = {math.PR},
  note          = {arXiv:2603.12767 [math.PR]},
  doi           = {10.48550/arXiv.2603.12767},
  url           = {https://arxiv.org/abs/2603.12767}
}

@article{Tarpey1994,
  author  = {Tarpey, Thaddeus},
  title   = {Two Principal Points of Symmetric, Strongly Unimodal Distributions},
  journal = {Statistics \& Probability Letters},
  year    = {1994},
  volume  = {20},
  number  = {4},
  pages   = {253--257},
  doi     = {10.1016/0167-7152(94)90012-4}
}

@article{VogelLachout2003,
  author  = {Vogel, Silvia and Lachout, Petr},
  title   = {On Continuous Convergence and Epi-Convergence of Random
             Functions. Part {II}: Sufficient Conditions and Applications},
  journal = {Kybernetika},
  year    = {2003},
  volume  = {39},
  number  = {1},
  pages   = {99--118},
  url     = {https://dml.cz/handle/10338.dmlcz/135510}
}
\endgroup

\end{document}